\documentclass[12pt]{article}

\usepackage{amsmath,amsfonts,amssymb,latexsym,graphicx,hyphenat}

\usepackage{caption}
\usepackage{subcaption}

\newtheorem{theorem}{Theorem}[section]
\newtheorem{corollary}{Corollary}[section]
\newenvironment{proof}[1][Proof]{\textsc{#1.} }{\ \rule{0.5em}{0.5em}}

\numberwithin{equation}{section}
\def\be{\begin{equation}}
\def\ee{\end{equation}}
\def\bq{\begin{eqnarray}}
\def\eq{\end{eqnarray}}
\def\beq{\begin{eqnarray}}
\def\eeq{\end{eqnarray}}

\begin{document}
\title{\textsc{Friedmann-Lema\^itre universes and their metamorphoses}}
\author{\Large{\textsc{Spiros Cotsakis}}$^{1,2}$\thanks{\texttt{skot@aegean.gr}}\\
$^{1}$Clare Hall, University of Cambridge, \\
Herschel Road, Cambridge CB3 9AL, United Kingdom\\  \\
$^{2}$Institute of Gravitation and Cosmology,  RUDN University\\
ul. Miklukho-Maklaya 6, Moscow 117198, Russia}
\date{May 2025}
\maketitle
\newpage
\begin{abstract}
\noindent We analyze the dynamics of the Friedmann-Lema\^itre universes taking into account the different roles played by the fluid parameter and the cosmological constant, as well as the degenerate character of the equations. We find that the Friedmann-Lema\^itre system reduces to four qualitatively inequivalent  normal forms and write down the sets of all stable perturbations that may result (the `versal unfoldings'). These sets are of small codimension up to three. We then describe all possible parameter-dependent solutions and their transfigurations to other forms during evolution through the  bifurcation sets, these are also fully described. This analysis leads to a picture of cosmological evolution determined by new parameters related to codimension which are zero in  standard cosmology. The emerging versal solutions are all free of singularities, while other properties of them  are also discussed.
\end{abstract}
\newpage
\tableofcontents
\newpage
\section{Introduction}
The standard model of physical cosmology is based on the \emph{Friedmann-Lema\^{i}tre} (hereafter `FL') equations which describe different evolutionary aspects of the universe and enable to host a great variety of further physical effects, cf. e.g., \cite{weinberg1}-\cite{bau} and refs. therein. There are many cosmological models that  can be built as solutions of the FL equations, all homogeneous and isotropic, and all have a duality associated with them, either they are `evolving', or they are in a statistically steady state. Many of the FL models currently favored by observations  started from a big bang,  a thermal state of infinite curvature and temperature, and have a phase of acceleration either `initially' or currently (or both). There are also models that may be said to have `started' near the Einstein static solution, or with quasi-de Sitter space as in many inflationary theories.  Using the FL equations one may also proceed to perturb a given exact solution and study the evolution of inhomogeneities, in turn leading to a theory of small fluctuations for the observed structures either in the microwave sky, or in baryonic matter, dark matter, etc, also having acquired  importance for and from the early stages in the evolution of the universe.

The mathematical notion of \emph{stability} of solutions of the FL equations in its many forms plays a central role in deciding about the plausibility and direct physical relevance of many of the cosmological ideas above throughout the cosmological development over a period of more than 100 years (cf. \cite{cy22}-\cite{cmm23} and refs. therein). Stability studies inside the standard model include not only those of de Sitter space, inflationary solutions,  the Einstein static universe, or the study of perturbations, but also various model solutions outside it (cf. the references). If a cosmological model is based on an unstable solution of the FL equations, then its plausibility as a candidate physical cosmology is greatly lost, and one is usually driven to seeking other models with better stability properties to base their studies of further physical effects. The  definition and application  of stability as is used in cosmology involves notions that belong to the  so-called `hyperbolic' theory of dynamical systems, cf. \cite{hs}-\cite{du0}. In this sense, stability in its various usages in cosmology usually  refers to linearized stability, a particular form of stability which states that an equilibrium solution of the FL equations (or other cosmological dynamical systems for that purpose) is linearization stable if all eigenvalues of the linear part of the system  have negative real parts. Otherwise the solution is unstable (cf. e.g., \cite{hs}, chap. 9).

Studied this way, the FL equations lead to cosmologies with usually complicated hyperbolic equilibria, and impose two important conditions on all universe models they include: firstly, cosmological models become qualitatively equivalent to their linear parts, and secondly, nearby systems behave like their unperturbed counterparts. In other words, FL cosmologies are usually  treated as being \emph{structurally stable}. This view leads to only three possible behaviours associated with any cosmological solution studied this way, stable, unstable, or saddle-like (a particular kind of unstable situation that also includes some stable orbits). That is, either there is nothing new and the universe returns to its unperturbed initial state, or the system diverges and is ripped apart, or there is an unstable mixture of the two. In such a situation the universe is essentially always `trapped in itself' being unable to pass to another, qualitatively different, state in the course of its evolution.

There is a great difference between the above treatment and the main message of this paper, which is that the FL equations are structurally \emph{unstable}. We propose that the FL equations describe phenomena which are essentially nonlinear and `dispersive' and  must be treated in a way that this feature does not disappear altogether, as in a structurally stable approach, but persists. In particular, the observed structures in the universe could never have arisen  if gravity was a structurally stable phenomenon, and this has important implications for cosmology.

In this paper, we highlight the different roles played by the two parameters, the fluid parameter  $\gamma$, and the cosmological constant $\Lambda$, and study the FL equations as a system that is essentially structurally unstable. We show that these equations belong to the most degenerate kind of dynamical systems in the plane having linear part that has two zero eigenvalues, and subsequently  classify the different kinds of degeneracy into four main types allowed by the combined effect of $\gamma$ and the symmetries of the FL equations, in particular, the presence or absence of a time-symmetry. We then embed the resulting systems in four inequivalent topological normal forms, the so-called  `versal families' - sets of equations which contain complete information about the set of all stable perturbations of the original equations - and discuss the corresponding bifurcation diagrams. These diagrams represent the solution sets of all possible perturbations of the original FL equations, and clearly depict all possible transfigurations - or metamorphoses - of the solutions. A sort discussion of the main physical effects associated with the constructions of the final `versal cosmologies' is also provided.

The plan of this paper is as follows. In Section \ref{rev}, we review some of the standard solutions of the FL equations and present a critical discussion of their linearized stability properties. In Section \ref{fl-bifn1}, we motivate the need for a bifurcation approach to the FL equations and also discuss the different nature of the parameters involved in these equations. In Sections \ref{cusp}-\ref{gam}, we discuss in more detail the  FL normal forms and present their versal unfoldings. The FL bifurcation diagrams are discussed in Section \ref{bifn-dia}, and we further discuss certain aspects of these results in the last Section of this paper.

We direct the reader to the introductions in the beginnings of the  sections of this paper for a more detailed description of their contents. Some of the most important points of this work are: 1) the two theorems in Sections \ref{flnf}, \ref{flva}, 2)  the role of the fluid parameter $\gamma$ as discussed in Section \ref{roleg} and the meaning of Fig. \ref{t-gamma}, and 3) the discussion of the final bifurcation diagrams in Section \ref{b-dia} including the elucidation about the rough character of the versal solutions in Section \ref{rough}. In addition, the reader may find the discussion in Section \ref{dis} useful.

\section{Brief review of the FL equations and their basic solutions}\label{rev}
In this Section, we first review some important solutions of the FL equations for different ranges of values of the constants $\gamma,\Lambda$, that are typically obtained when seeking analytic expressions and closed-form solutions of the equations. We then motivate the use of dynamical system formulations for the FL equations and  discuss the important notion of linearized stability, that is the stability  of the solutions by looking at the linearized equations around some equilibrium solution. We finally discuss the importance of the $\gamma,\Lambda$ `constants' and the limitations of such `hyperbolic-equilibrium' formulations for the FL equations. This discussion naturally leads to  the necessity  of a bifurcation treatment of the FL cosmology.

\subsection{Some known FL solutions}\label{fl-solns}
The  FL equations can be directly obtained from the Einstein equations for a Friedmann metric with scale factor $a(t)$ as a function of the proper time $t$,  a perfect fluid source with density $\rho$ and pressure $p$, and a cosmological constant $\Lambda$. In their `primitive' form they read,
\beq
&&\frac{\ddot{a}}{a}+\frac{1}{6}(\rho+3p)-\frac{1}{3}\Lambda=0\label{fl1}\\
&&\left(\frac{\dot{a}}{a}\right)^2+\frac{k}{a^2}-
\frac{1}{3}\rho-\frac{1}{3}\Lambda=0\label{fl2}\\
&&\dot{\rho}+3\frac{\dot{a}}{a}(\rho+p)=0\label{fl3}.
\eeq
Here we have set $8\pi G=c=1$, and $k=0,\pm 1$ denotes the normalized constant curvature of the spatial 3-slices.  These are the standard equations of isotropic and homogeneous cosmology, while the case $\Lambda=0$ corresponds to the \emph{Friedmann} equations.

As discussed in the Introduction, the standard model of cosmology  is based on the FL equations (\ref{fl1})-(\ref{fl3}), a fact which shows the importance attached to these equations as clearly demonstrated by the abundance of related sources devoted to their study and physical effects based on them. In fact, the FL equations  and their solutions also play an important role in any attempt to go \emph{beyond} the standard cosmological model, cf. e.g., the  review \cite{cy22} where the development of FL-cosmology is placed in a more general scheme of things cosmological.

There are two main methods to study the FL equations  (cf. e.g., \cite{weinberg2}, pp. 8-55, \cite{pu09}, pp. 138-50, \cite{we97}, pp. 51-67). The first  aims to determine analytic expressions of the solutions. It usually starts by  noting that  these equations are \emph{not} independent, for instance, Eqns. (\ref{fl2}), (\ref{fl3}) imply (\ref{fl1}). Then the usual strategy  consists of the following two steps: (1) study the continuity equation (\ref{fl3}) for the equation of state $p=(\gamma -1)\rho$ ($\gamma$ is  the fluid or state parameter, related to the square of the sound speed) by writing it in the `integrated form' $\rho=C a^{-3\gamma}$ ($C$ is an integration constant), valid for any constant $\gamma$ (i.e., the fluid is barotropic), and (2) solve the Friedmann equation (\ref{fl2}) to determine $a(t)$. The following important solutions - among many others - of the FL equations have been obtained using this method:
\begin{enumerate}
  \item The $k=+1$ Einstein static universe with nonzero $\Lambda$.
  \item The $k=0$, empty de Sitter space with nonzero $\Lambda$.
  \item The $k=0$ Friedmann universe with zero $\Lambda$.
  \item The $k=-1$ Milne universe  with zero $\Lambda$, and $p=\rho=0$.
\end{enumerate}
From these solutions, a number of other  universes can be obtained when the constants $\gamma, \Lambda$ take special values or are constrained in certain intervals. For future reference and comparison, we note the following important classes of solutions to the FL equations associated with specific ranges of the fluid parameter $\gamma$:
\begin{itemize}
 \item  $\gamma\geq 1$: this is the standard cosmology case (dust and radiation models belong here).
\item $\gamma<1$: this is the standard negative pressure case. An expanding universes releases energy with expansion.
\item $\gamma\in(0,2/3)$: Closed oscillating universes where oscillations slowly decay leading to closed static stable universes.
\item $\gamma=2/3$: such a universe contains a dark component that behaves as a curvature-like fluid.
\item $\gamma=0$: this is an interesting situation of an expanding universe of constant density. This includes a flat, steady state model, and also the inflationary universe.
\item $\gamma<0$: tension (as negative pressure) exceeds the energy density leading to models where the density increases/decreases with expansion/contraction.
\end{itemize}
In addition, the \emph{Einstein-de Sitter} solution is found when $\gamma=1$, the case of `dust', or \emph{cold matter}. The radiation-filled, also known as \emph{hot matter}, universes having $\gamma=4/3$ exist for all values of $k$, while further special solutions exist for special values of $k, \gamma,$ and/or $\Lambda$, for instance, the \emph{Eddington-Lema\^itre} models, and the \emph{Lema\^itre} universes which have transient stages near the Einstein static or de Sitter solutions (cf. the references). Solutions to the FL equations are called \emph{ordinary (resp. dark)} matter-energy solutions when certain extra inequalities are satisfied, namely, $\rho+3p>0$ (resp. $<0$). Specifically, one sometimes refers to solutions of the FL equations satisfying the inequality constraint $\gamma<2/3$ as \emph{dark energy}. For example, one refers to a fluid with $\gamma=-2/3$ as a particular case of \emph{phantom energy} (cf. e.g., \cite{weinberg2}, p. 55).

We  note that from the Raychaudhuri equation (\ref{fl1}) we may obtain a useful criterion for the existence of \emph{accelerating} expanding solutions to the FL equations. Introducing the deceleration parameter $q=-\ddot{a}/(aH^2)$, and dividing Eq. (\ref{fl1}) by $H^2$,  one has accelerating (i.e., bending-up $\ddot{a}>0$) solutions with $q<0$, satisfying Eq. (\ref{fl1}) in the form,
\be
q=\frac{3\gamma-2}{2}\,\Omega_M -\Omega_\Lambda,
\ee
where the standard definitions for the \emph{energy parameters} $\Omega_M, \Omega_\Lambda$ for matter and cosmological constant respectively are used. So when $\Lambda=0$,  a dominating cosmic fluid inducing acceleration requires $3\gamma-2<0$, in compliance with the dark energy constraint. On the other hand, there are obvious accelerating solutions having only   `vacuum' $\Lambda>0$, the simplest form of dark energy. Inflationary dynamics, now part of the standard model of cosmology, can also be described using such accelerating solutions, cf. \cite{pu09}, chap. 8.

\subsection{Dynamical systems formulations}
Another  approach to the study of the FL equations (\ref{fl1})-(\ref{fl3}) is based on the important idea that certain of the known solutions, for instance the de Sitter, Einstein static universe (ESU), Milne, flat solutions, may be regarded as \emph{equilibrium solutions} in a dynamical system formulation, and so each one of those become a kind of `organizing center' for the evolution of all others.

A little reflection shows that the equations (\ref{fl1})-(\ref{fl3}) can be written as an autonomous evolution  system of two first-order ordinary differential  equations with an extra algebraic constraint. For the dynamical variables $(H,\rho)$, where $H=\dot{a}/{a}$ is the Hubble parameter, one obtains the following  dynamical system in terms of the $(H,\rho)$-variables,
\beq
\dot{H}&=&-\frac{3\gamma-2}{6}\rho+\frac{\Lambda}{3}-H^2 \label{ds1}\\
\dot{\rho}&=&-3\gamma H\rho, \label{ds2}
\eeq
and the algebraic constraint
\be
a^2=\frac{3k}{\rho+\Lambda-3H^2}. \label{ds3}
\ee
This form of the FL equations will be discussed in more detail later. Of course, this reduction of Eqns. (\ref{fl1})-(\ref{fl3}) is highly  non-unique, other choices of variables are possible and lead to dynamical systems of different forms, cf. e.g., \cite{we97}, or, \cite{pu09}.

However, the form obtained in Eqns. (\ref{ds1})-(\ref{ds2}) will prove most useful for our purposes, a  reason for this being that this form of the FL equations separates the unknowns (i.e., the phase space coordinates $H,\rho$) from the parameters $\gamma, \Lambda$. The important role that the fundamental Friedmann equation (\ref{ds3})  plays in understanding  the existence of rough (i.e., non-smooth, or even discontinuous) solutions to the `versal' versions of the  system (\ref{ds1})-(\ref{ds3}) will be discussed in Section \ref{rough}.

A common aspect of different dynamical systems formulations of the FL equations usually considered in the literature (either in standard relativistic cosmology, or in modified gravity theories, string cosmology, etc), cf. e.g., \cite{bogo,we97,cop18}, is that they are all based on the use of the \emph{hyperbolic theory} of dynamical systems cf. eg., \cite{hs,arny78,ap92}. These tools are of a qualitative, topological nature, and are useful for two reasons:
\begin{itemize}
\item They implement the existence of organizing centers associated with the dynamics of the FL equations referred to above as \emph{equilibrium solutions} of these equations.
\item They help decide about the \emph{linearized stability} of these equilibria, and hence offer a very valuable approach to one of the most mysterious and central problems in theoretical cosmology, namely,  the long-term behaviour of solutions of the FL equations near their equilibria.
\end{itemize}
For example, these `newer' methods have been applied to various cosmological contexts in this sense of proving linearized  stability and related results:  to complement and extend older works about the (in-)stability of the Einstein static universe, or the stability of de Sitter space, or deduce novel results about the behaviour of entropy, the possibility of singularity-free solutions, as well as other important properties of the FL equations, but also in the more general  Bianchi cosmologies generalizing the FL solutions (cf. e.g., \cite{harrison67}-\cite{cmm23}, the reviews \cite{we97,cop18,bogo}, and Part B, Sect. i, of \cite{cy22}, and refs. therein).

\subsection{Linearized stability}\label{lin-stab}
It is interesting that the Einstein static universe and de Sitter space appear as \emph{equilibrium solutions} of the FL equations (\ref{fl1})-(\ref{fl3}) in their dynamical system formulation (\ref{ds1})-(\ref{ds2}). The system (\ref{ds1})-(\ref{ds2}) has two equilibria:
\begin{enumerate}
\item \textbf{EQ-1a,b: Flat and empty de Sitter space.}
    \begin{equation}\label{ds-eq}
    \begin{split}
    H&=+\sqrt{\frac{\Lambda}{3}},\\
    H&=-\sqrt{\frac{\Lambda}{3}},
    \end{split}
    \end{equation}
with the $+$ (resp. $-$) sign for EQ-1a (resp. b), $\Lambda>0$, any $\gamma$, $\rho=0$, and $k=0$ .
\item \textbf{EQ-2: Einstein static universe.}
       \be\label{esu-eq}
    H=0,\quad\rho=\frac{\Lambda}{3\gamma-2},
    \ee
for any $\Lambda$, and $\gamma\neq 2/3$.
\end{enumerate}
These equilibria are of course very well-known and exist only when $\Lambda\neq 0$. In the case of the ESU, to keep $\rho\geq 0$ we require either Case 1: $\Lambda\geq 0,\gamma>2/3$, or Case 2: $\Lambda< 0,\gamma<2/3$. These equilibria exist together only in case 1, whereas in Case 2 only the ESU equilibrium exists. When $\gamma=2/3,$ we only have the  de Sitter space as an equilibrium.

To simplify the subsequent analysis, we shall move the equilibria EQ1, EQ2 to the origin and write  the equations in a form that is easier to work with. This is simply done by a linear change of the dynamical variables. For EQ1, we set $h=H\mp \sqrt{\Lambda/3}$ and   the system (\ref{ds1})-(\ref{ds2}) `splits off' the linear from the nonlinear part and gives,
\be \label{ds-sys}
\left(
  \begin{array}{c}
    \dot{h} \\
    \dot{\rho} \\
  \end{array}
\right)
=
\left(
  \begin{array}{cc}
    \mp\sqrt{\frac{\Lambda}{3}} & -\frac{(3\gamma-2)}{6} \\
    0 & \mp 3\gamma\sqrt{\frac{\Lambda}{3}} \\
  \end{array}
\right)
\left(
  \begin{array}{c}
    h \\
    \rho \\
  \end{array}
\right)
+
\left(
  \begin{array}{c}
    -h^2\\
-3\gamma h\rho \\
  \end{array}
\right).
\ee
The first term on the right of this equation is the linear part of the `FL vector field' defined by Eqns. (\ref{ds1})-(\ref{ds2}) at the origin (the de Sitter space equilibrium).  Similarly, for EQ2 and for $\gamma\neq 2/3$, we set $w=\rho-2\Lambda /(3\gamma-2)$, and obtain the following system for the new variables $(H,w)$ with the linear part evaluated at the origin (i.e., the ESU),
\be \label{esu-sys}
\left(
  \begin{array}{c}
    \dot{H} \\
    \dot{w} \\
  \end{array}
\right)
=
\left(
  \begin{array}{cc}
    0 & -\frac{(3\gamma-2)}{6} \\
    -\frac{6\gamma\Lambda}{3\gamma-2} & 0 \\
  \end{array}
\right)
\left(
  \begin{array}{c}
    H \\
    w\\
  \end{array}
\right)
+
\left(
  \begin{array}{c}
    -H^2\\
-3\gamma Hw \\
  \end{array}                                                                                                                                       \right).
\ee
Both systems (\ref{ds-sys}), (\ref{esu-sys}) have linear parts (given by the first terms on their right-hand-sides) equal to the linear parts of the FL vector field at the corresponding  origins, and so they are obviously equivalent to the FL equations (\ref{ds1})-(\ref{ds2}) near their respective origins. (For the equilibrium EQ1, we shall only treat the case with the minus sign in front of the diagonal elements of the Jacobian in Eq. (\ref{ds-sys}), the other case is very similar.)

To decide on the linearized stability of the de Sitter space and of the Einstein static universe equilibria, one merely needs to calculate the eigenvalues of the Jacobians evaluated at these two equilibria of the systems (\ref{ds-sys}), (\ref{esu-sys}) respectively. In this way, one thinks of the linear parts of the vector fields in (\ref{ds-sys}), (\ref{esu-sys})  as approximating the FL vector field defined by Eqns. (\ref{ds1})-(\ref{ds2}) near the origin in each case.

Under this assumption one immediately concludes that for $\Lambda, \gamma >0$,  both eigenvalues of the Jacobian in (\ref{ds-sys}) are indeed negative, thus de Sitter space is a (nonlinear) sink attracting all nearby solutions, whereas the eigenvalues of the Jacobian in (\ref{esu-sys}) are $\pm\sqrt{\Lambda\gamma}$, thus the ESU is a saddle. Since nonlinear sinks behave like linear sinks, nearby solutions approach the origin exponentially and therefore   de Sitter space is a stable solution in the linear approximation. Similarly, since saddles are unstable, it follows that the ESU is an unstable solution of the FL equations in the linear approximation (this conclusion was already known by examining the system (\ref{ds1})-(\ref{ds2}) directly,  and clarified the usual meaning of the phrase `the ESU is unstable', cf. e.g., \cite{gibbons87} and refs. therein). There is also another case, not usually considered in linearized studies about the issue of the  `stability of the ESU', that is when $\Lambda<0$. Following our discussion after Eq. (\ref{esu-eq}), when $\Lambda<0$ we may distinguish two further subcases of Case 2: In the Case 2a: $\Lambda< 0,\gamma<0$, the ESU is again a saddle like before. However, in the Case 2b: $\Lambda< 0,0<\gamma<2/3$, we  have a pair of nonzero purely imaginary eigenvalues and so this subcase corresponds to a situation that has to remain undecided in the linear approximation. For as it follows from Eq. (\ref{esu-sys}) the origin is `linearly a centre' and this leads to the question of the \emph{structural} stability of the Eq. (\ref{esu-sys}), namely, whether this equation is stable with respect to perturbations which will generally alter the structure of the phase space of the problem in other nontrivial ways (for instance if it is a centre, the phase curves are ellipses in this case, and the presence of higher-than-linear `perturbing' terms in (\ref{esu-sys}) may be expected to completely destroy the ellipses).

The presence of the original second-order terms in the Eqns. (\ref{ds-sys}), (\ref{esu-sys}) is in a sense just `symbolic' here, we have simplified (in the sense of reaching the normal form)  only the linear parts by moving the equilibria to the origin (it so happens that the linear part so obtained is already in Jordan form). The most important problem to be addressed  in this paper is  to  fully understand the nature of the second-order terms and the possible influence of  higher-order terms in the systems (\ref{ds-sys}), (\ref{esu-sys}). In fact, this problem is not really relevant for linearized stability studies of the FL equations (and vice versa!) since for such studies one is only interested exclusively in the \emph{linear} part of the vector field.

\subsection{Some preliminary comments}
There are two somewhat concealed aspects which naturally arise in connection with the FL equations  (\ref{ds1})-(\ref{ds2}) (or (\ref{ds-sys}), (\ref{esu-sys})) that deserve special attention, we briefly describe them here. Firstly, as we discussed in Section \ref{fl-solns}, interesting and physically relevant solutions of the FL equations may occur for a wide variety of different values of $\gamma, \Lambda$. Therefore,  since  $\gamma, \Lambda$ are not really known exactly and also acquire different values during cosmological history satisfying relevant physical constraints, it appears more reliable not to restrict their ranges \emph{ad hoc} as in searches for linearized stability, but instead treat them as \emph{parameters}, that is continuous variables independent of the time, i.e., $\dot{\Lambda}=\dot{\gamma}=0$. This will allow us to carefully monitor the influence  they might have on cosmological phenomena as these parameters change. It will also lead us to discover the true relevance of \emph{nonlinear} stability for cosmology in the context of such a complicated set of dynamical equations as the FL system.

Secondly, to apply the theory of hyperbolic systems (leading to linearized (in-) stability) to the FL equations presupposes that for the stated ranges of the parameters $\gamma,\Lambda$ near the equilibria, these equations behave like their linear parts (cf. the linear parts in the Eqns. (\ref{ds-sys}), (\ref{esu-sys})). This is of course an extra assumption whose validity entirely depends on whether or not the assumed hyperbolicity of a given equilibrium is violated for some nearby value of the parameters due to the possible presence of \emph{degeneracies}. If it does it dramatically affects the conclusion, and  the whole approach of linearized stability and  general application of hyperbolic theory may not be an appropriate tool  to decide on the stability of the given equilibrium (an equilibrium is \emph{non-degenerate} of 0 is not an eigenvalue).

We re-emphasize the subtle fact that having to examine some special system that arose in some area of application (such as the FL equations presently) is quite different from  developing the general hyperbolic theory of dynamical systems - these are \emph{already assumed} to be hyperbolic for this purpose, cf. \cite{hs,arny78,ap92}. Dynamical equations that arise in gravitational studies very rarely contain hyperbolic equilibria without extra assumptions (like above), and when they do they represent an exception that proves the rule. To deal with  nonhyperbolic cases a whole branch of mathematics is available, namely, bifurcation theory with all of its methods and subfields.

We shall employ some background about the dynamics and geometry of degenerate systems  and bifurcations as we go along with the analysis of the FL equations at different stages in this work. For sources on bifurcation theory at various levels, we refer the reader to \cite{du0}-\cite{stro}, while for an introduction to bifurcation theory notions and results in a gravitational setting, see \cite{cot23}-\cite{cot25b}.

\section{The FL equations as a bifurcation problem}\label{fl-bifn1}
In this Section, after a discussion of the different nature of the two distinguished parameters of the problem, $\gamma,\Lambda$, we write down the four inequivalent forms of the FL equations explicitly exhibiting the different possible types of degenerate behaviour associated with these equations. This leads to a statement of the normal forms of the FL equations as well as their versal unfoldings, the topologically inequivalent differential systems which contain all possible perturbations of the original FL equations.

\subsection{Nature of $\gamma$ and $\Lambda$}
To begin appreciating the different roles played by the two distinguished parameters $\gamma,\Lambda$ of the FL equations (\ref{ds-sys}), (\ref{esu-sys}), we may consider one of them as fixed, vary the other, and check for possible degeneracies and/or multiplicity of steady state solutions. Instead of restricting the parameters to obtain complicated hyperbolic equilibria, we ask: is there a parameter value such that the Jacobians in Eqns. (\ref{ds-sys}), (\ref{esu-sys})) have some zero eigenvalue? (We note that if the answer to this question was negative, then hyperbolic analysis would be \emph{the} tool to use.) If the answer is positive, the next step is to look for possible variations in the number of equilibria as the parameters change. We prove the following
\begin{theorem}\label{lm}
The parameters $\gamma,\Lambda$ which appear in  the FL equations (\ref{fl1})-(\ref{fl3}) (or equivalently (\ref{ds1})-(\ref{ds3})) are as follows\footnote{We note that there is another kind of bifurcation parameter, a \emph{modular coefficient} (or simply  `moduli'). This is  one that is related to inequivalent bifurcation diagrams. In general, in a bifurcation problem the presence of moduli complicates matters because for each value of a moduli corresponds a different, inequivalent diagram. We shall meet such  moduli parameters because they will be present in all versal unfoldings below. The fluid parameter $\gamma$ is not a modular coefficient, as it does not appear in the versal families. However, it is not an unfolding parameter either, it is just a \emph{distinguished} parameter of the problem, i.e.,  a parameter that appears in the original equations (cf. \cite{golu2} for this terminology).}:
\begin{enumerate}
\item $\gamma$ is not a bifurcation parameter, and
\item $\Lambda$ is an unfolding parameter.
\end{enumerate}
\end{theorem}

\begin{proof}
There are two cases to consider as one of the parameters $\gamma$ or $\Lambda$ varies and the second remains fixed.

\textbf{A. The singular case: $\gamma$ varies, $\Lambda$ is fixed.}

We examine the two Eqns. (\ref{ds-sys}), (\ref{esu-sys}) separately.  Using the FL equations in the form (\ref{ds-sys})  it is straightforward to see that when $\gamma=0$ and $\Lambda$ is fixed and nonzero, the eigenvalues are $\mp\sqrt{\Lambda/3}, 0$, so there is always a zero eigenvalue. This means that the assumptions of the  implicit function theorem are violated and so the critical point in question is a singularity (i.e., not a regular value according to that theorem). It remains to look at the possible multiplicity of solutions at $\gamma=0$ (if $\gamma$ is to quantify as a bifurcation point). Since by the constraint (\ref{ds3}) we have that  $3H^2\neq \rho+\Lambda$, using the Eqns. (\ref{ds1}), (\ref{ds2}) we either find de Sitter space or the ESU in the forms (\ref{ds-eq}), (\ref{esu-eq}). Then when $\gamma\neq 0$, from Eq. (\ref{ds2}) we have that either $H=0$, in which case we get the equilibrium (\ref{esu-eq}), or $\rho=0$ and we get the equilibrium (\ref{ds-eq}). So $\gamma$ cannot be a bifurcation parameter for the Eq. (\ref{ds-sys}).

For the case of the system (\ref{esu-sys}), the proof proceeds as before with the only difference that when $\gamma=0$ and $\Lambda$ is fixed, the linear part is nilpotent (a double zero eigenvalue).

Therefore in both cases, no change in the number of equilibria is found. Hence, although this is a \emph{singular} case in the sense of violating the regularity assumptions of the implicit function theorem (Jacobian degenerate), since the number of equilibria remains fixed and equal to two as the parameter $\gamma$ changes and passes through zero (for fixed $\Lambda$), there is no bifurcation.

\textbf{B. The bifurcation case: $\Lambda$ varies, $\gamma$ is fixed.}

We first note the  fact that when $\Lambda=0$ and $\gamma$ fixed, the Eqns. (\ref{ds-sys}), (\ref{esu-sys}) acquire an identical form, namely,
\be \label{deg-sys}
\left(
  \begin{array}{c}
    \dot{H} \\
    \dot{\rho} \\
  \end{array}
\right)
=-\frac{(3\gamma-2)}{6}
\left(
  \begin{array}{cc}
    0 & 1\\
    0 & 0 \\
  \end{array}
\right)
\left(
  \begin{array}{c}
    H \\
    \rho \\
  \end{array}
\right)
+
\left(
  \begin{array}{c}
    -H^2\\
-3\gamma H\rho \\
  \end{array}
\right).
\ee
In other terminology \cite{thom}, Eq. (\ref{deg-sys})  means that the FL equations have a \emph{unique organizing centre}. This fact allows for a unified bifurcation treatment of the full FL equations (\ref{ds1})-(\ref{ds3}).

To examine for bifurcating behaviour, as a first step one counts and identifies  the multiplicity of equilibrium solutions as $\Lambda$ changes.
When  $\gamma=2/3$, as $\Lambda$ passes through zero the number of solutions of the form (\ref{ds-eq}) jumps from zero when $\Lambda<0$ to one at $\Lambda=0$ (the trivial solution $H=\rho=0$),  to two for $\Lambda>0$ (the $\pm$ solutions for de Sitter). (The behaviour of the Eq. (\ref{deg-sys}) \emph{itself} is a different story, see below.) When $\gamma\neq 2/3$, the ESU equilibrium from Eq. (\ref{esu-eq}) will have to be  added to the previous counting for the $\Lambda\neq 0$ cases. Also the FL equations in the form (\ref{deg-sys}) represent  a fully degenerate two-dimensional system with linear part that has two zero eigenvalues. This is the nilpotent 2-matrix when $\Lambda\neq 0$, and when $\gamma= 2/3$ the linear part is the zero matrix (this is the only other  case with a double zero eigenvalue for planar systems). This implies that for different $\gamma$-values bifurcation diagrams will occur \emph{with modified codimension}: different numbers of  parameters will be needed to unfold the singularities associated with the different linear parts of the FL equations.

Therefore in all cases of fixed $\gamma$, we find that $\Lambda$ is a bifurcation parameter for the FL equations because it satisfies both conditions for this purpose, namely, giving both a singularity (degenerate Jacobian) at the origin,  and a bifurcation point at $\Lambda=0$. This completes the proof of the Theorem \ref{lm}.
\end{proof}

The role of the fluid parameter $\gamma$ in this problem will also  be treated later  where we shall show (as we have already indicated) that the bifurcation diagrams corresponding to different values of $\gamma$ are topologically inequivalent: for example, the case $\gamma=2/3$ compared to $\gamma \neq 2/3$, cf. Section \ref{roleg}. Despite the different roles played by $\gamma$ and $\Lambda$ as per the  Theorem \ref{lm} and the bifurcation results of the following Sections of this paper, there is a physical ambiguity still remaining because of the well-known fact that the case $\gamma=0$ corresponds to just a cosmological constant. Although it is a trivial algebraic fact to associate one to the other (cf. e.g., \cite{weinberg2}, pp. 43-4), and so it would seem  somewhat pedantic to discuss it any further,  this standard identification generally corresponds to a non-dynamical, or `static',  equivalence of the vacuum ($\gamma=0$)-fluid state and the cosmological constant. However, since these two states usually appear in a cosmological context to be related to equilibrium states of the system, the question naturally arises whether the two states, the $\gamma=0$-fluid state and any nonzero-fixed-value $\Lambda$-state, are also equivalent when considered as two different transient states  approached by the system dynamically?  This is not a trivial question, and
due to its various intricacies, we shall leave a more detailed discussion of this problem for the future.

\subsection{Inequivalent degenerate forms of the FL equations}

The system (\ref{deg-sys}) is the core degenerate system associated with the FL equations, as it contains all degenerate subcases of these equations that need to be considered in order to arrive at a full picture of the behaviour of all possible perturbations of these equations.

To enumerate the various cases that can occur, it is necessary to simplify it as much as possible, eventually leading to all possible topological normal forms that this system can take. This will be completed in the next Section. In the rest of this Section, we shall present a preliminary discussion which serves as an introduction to this problem.

\subsubsection{Time-symmetric Case: $\gamma\neq\pm 2/3$ (includes de Sitter and the ESU)}\label{t-symm case}
For the case $\gamma\neq 2/3$, if we   redefine the time,
\be\label{red-t}
t\to -\frac{(3\gamma-2)}{6}t,
\ee
(we shall use a $'$ in the place of the $\cdot$ for differentiation), and introduce the  notation,
\be
a=\frac{6}{3\gamma-2},\quad b=\frac{18\gamma}{3\gamma-2},
\ee
we arrive at the following equivalent form of the system (\ref{deg-sys}),
\begin{equation}\label{deg-sys1}
  \left(
  \begin{array}{c}
    H' \\
  \rho' \\
  \end{array}
\right)
=
\left(
  \begin{array}{cc}
    0 & 1\\
    0 & 0 \\
  \end{array}
\right)
\left(
  \begin{array}{c}
    H \\
    \rho \\
  \end{array}
\right)
+
\left(
  \begin{array}{c}
    aH^2\\
b H\rho \\
  \end{array}
\right).
\end{equation}
The system (\ref{deg-sys1}) is (like (\ref{deg-sys})) $(\mathbb{Z}_2 +t)$-\emph{equivariant}, meaning that it is invariant under the symmetry,
\be\label{t-symm}
H\to -H,\quad \rho\to \rho,\quad t\to -t.
\ee
We shall show that exploiting this symmetry, the system (\ref{deg-sys1}) leads to a final versal unfolding which  is also manifestly $(\mathbb{Z}_2 +t)$-equivariant. We shall call the appearance of  this symmetry  in the final topological normal form a \emph{`time-symmetric case'}.

\subsubsection{Time-symmetric Case: $\gamma=-2/3$ (phantom energy)}\label{-2/3}
In the case when $\gamma=-2/3$ in Eq. (\ref{deg-sys}),  the system can also be readily written in a time-symmetric form (redefining $t\to (2/3)t$), that is,
\be \label{deg-sys2}
\left(
  \begin{array}{c}
    \dot{H} \\
    \dot{\rho} \\
  \end{array}
\right)
=
\left(
  \begin{array}{cc}
    0 & 1\\
    0 & 0 \\
  \end{array}
\right)
\left(
  \begin{array}{c}
    H \\
    \rho \\
  \end{array}
\right)
-\frac{3}{2}
\left(
  \begin{array}{c}
    H^2\\
-2H\rho \\
  \end{array}
\right).
\ee
We shall say more above this system below.

\subsubsection{Time-asymmetric Case: $\gamma=2/3$ (curvature-like fluid)}\label{2/3}
On the other hand, when $\gamma=2/3$ (this excludes the ESU Eq. (\ref{esu-sys}) and retains only  Eq.  (\ref{ds-sys})), Eq. (\ref{deg-sys}) becomes,
\be \label{deg-sys3}
\left(
  \begin{array}{c}
    \dot{H} \\
    \dot{\rho} \\
  \end{array}
\right)
=-
\left(
  \begin{array}{c}
    H^2\\
2 H\rho \\
  \end{array}
\right).
\ee
The main difference between the systems (\ref{deg-sys}) (or (\ref{deg-sys1})) and (\ref{deg-sys3}) is that in (\ref{deg-sys3}) the linear part is zero making it a totally degenerate case. The other distinguished property of the  system (\ref{deg-sys3}) is that it is  invariant under the  $\mathbb{Z}_2$-\emph{symmetry},
\be\label{z2-symm}
H\to H,\quad \rho\to -\rho
\ee
but not the time symmetry $t\to -t$ .
We shall call any system in which the time reflection symmetry $t\to -t$ is absent, a \emph{time-asymmetric case}.

\subsubsection{Time-asymmetric Case: $\gamma\neq\pm 2/3$ (includes de Sitter and the ESU)}\label{t-symm case2}
We note here the final way  in which the degeneracies of the system (\ref{deg-sys}) become manifest. We shall show that when  $\gamma\neq\pm 2/3$, there is a second way to simplify the system (\ref{deg-sys1}), qualitatively \emph{in}equivalent to that noted in subsection \ref{t-symm case}, a form which as we shall prove is manifestly \emph{time-asymmetric}. This happens when the eventual versal unfolding does \emph{not} have the full $(\mathbb{Z}_2 +t)$ symmetry present in the original system (\ref{deg-sys}), but it is only $\mathbb{Z}_2$-equivariant.

\subsection{The FL  normal forms}\label{flnf}
All four inequivalent forms to which the system (\ref{deg-sys}) reduces to depend on $\gamma$ and lead to the following normal forms\footnote{A normal form is the simplest possible form that a dynamical system containing nonhyperbolic points can acquire. It is found by a series of smooth transformations starting from the original system. (Nonhyperbolic points are equilibria where the Jacobian of the linearized system evaluated there has some eigenvalues on the imaginary axis.) Bifurcation theory starts once a system is reduced to its normal form.}.
\begin{theorem}\label{deg-thm}
Each of the four primary inequivalent degenerate cases discussed above for the FL equation (\ref{deg-sys}) leads to the following normal forms classified according to the ranges of the modular parameter $\gamma$:
\begin{itemize}
\item[A.] Time-Asymmetric Case-1: $\gamma\neq\pm 2/3$. In this case (cf. subsection \ref{t-symm case2}), after a series of smooth changes  the system (\ref{deg-sys1}) reduces to the normal form,
\begin{equation}\label{deg-sys1a}
  \left(
  \begin{array}{c}
    x' \\
  y' \\
  \end{array}
\right)
=
\left(
  \begin{array}{cc}
    0 & 1\\
    0 & 0 \\
  \end{array}
\right)
\left(
  \begin{array}{c}
    x \\
    y \\
  \end{array}
\right)
+
\left(
  \begin{array}{c}
    0\\
Bxy +Cy^2 \\
  \end{array}
\right),
\end{equation}
where
\be \label{deg-sys1a1}
B=\frac{6(3\gamma+2)}{3\gamma-2},\quad C=\frac{12}{3\gamma-2}.
\ee
\item[B.] Time-Symmetric Case-1: $\gamma\neq\pm 2/3$.
In this case (cf. subsection \ref{t-symm case}), after a series of smooth changes  the system (\ref{deg-sys1}) reduces to the normal form,
\begin{equation}\label{deg-sys1b}
  \left(
  \begin{array}{c}
    x' \\
  y' \\
  \end{array}
\right)
=
\left(
  \begin{array}{cc}
    0 & 1\\
    0 & 0 \\
  \end{array}
\right)
\left(
  \begin{array}{c}
    x \\
    y \\
  \end{array}
\right)
+
\left(
  \begin{array}{c}
    0\\
Ax^3+Bxy \\
  \end{array}
\right),
\end{equation}
where
\be\label{deg-sys1a2}
A=-\frac{108\gamma}{3\gamma-2},
\ee
and $B$ as in Eq. (\ref{deg-sys1a1}).
\item[C.] Time-Symmetric Case-2: $\gamma=-2/3$. The system (\ref{deg-sys2}) reduces to the form given by,
\be \label{deg-sys2a}
\left(
  \begin{array}{c}
    \dot{x} \\
    \dot{y} \\
  \end{array}
\right)
=
\left(
  \begin{array}{cc}
    0 & 1\\
    0 & 0 \\
  \end{array}
\right)
\left(
  \begin{array}{c}
   x \\
   y \\
  \end{array}
\right)
+O(3).
\ee
This system can be brought to a normal form that is manifestly time symmetric.
\item[D.] Time-Asymmetric Case-2: $\gamma=2/3$. The system  (\ref{deg-sys})  reduces to the form given by the system (\ref{deg-sys3}).
\end{itemize}
The relations of the variables $(x,y)^T$ to the original ones will be given below.
\end{theorem}

\subsection{The FL versal unfoldings}\label{flva}
The fact that the FL equations contain the degenerate system (\ref{deg-sys}) means that these equations acquire inequivalent normal forms according to their degeneracy as discussed in the previous Subsection, Theorem \ref{deg-thm}. In the next four Sections, we shall employ an exhaustive list of all possible perturbations for the FL equations, the so-called  `versal unfoldings' of each one of the four basic FL normal forms\footnote{The notion of versal unfolding was  introduced by H. Poincar\'e around 1890. Bifurcation theory was also founded by Poincar\'e around the same time, and subsequently developed by A. A. Andronov, R. Thom, V. I. Arnold, and later by F. Takens, R. I. Bogdanov,  F. Dumortier, H. Zoladek, and many other brilliant mathematicians, cf. the refs.}.

The most complicated situation to deal with corresponds to the $\gamma\neq\pm 2/3$ time-symmetric case, however, all four cases taken together represent a vast structure of possible behaviours associated with the FL equations. In  all these  possibilities  the versality of the constructed families depends on second-order or third-order terms, as summarized in  cases A-D  of the  Theorem  below.
\begin{theorem}\label{deg-thm1}
Through different series of smooth changes and rescalings, the systems given in Theorem  (\ref{deg-thm}), lead to the following versal unfoldings in Cases A-D respectively.
\begin{enumerate}
\item[A.] \emph{The FL cusp.}

\item[B.] \emph{The saddle-focus-elliptic codimension-3 FL cubic.}

\item[C.] \emph{The saddle-focus codimension-2 FL cubic.}

\item[D.] \emph{The $\mathbb{Z}_2$-equivariant codimension-2 unfolding}
\end{enumerate}
\end{theorem}
This result together with Theorem  (\ref{deg-thm}) will be deployed in the next several Sections of this paper. Together they  contain the main information about the possible/allowed perturbations of the FL equations.
The  meaning of the symbols and technical terms that appear in these Theorems (\ref{deg-thm}), (\ref{deg-thm1}) will also be clarified in these Sections.

\section{The FL cusp}\label{cusp}
The purpose of this Section is to prove Theorem \ref{deg-thm}, Case A, as well as to develop the following result which settles Case A of Theorem \ref{deg-thm1}.
\begin{theorem}\label{deg-thm-cusp}
When  $\gamma\neq\pm 2/3$, the  FL equations in the form (\ref{deg-sys1a}) (cf. subsection \ref{t-symm case2}) versally unfold to
the topological normal form  described by the second-order Bogdanov-Takens versal unfolding, namely,
\begin{equation}\label{va-tot}
    \begin{split}
    \dot{x}&=y,\\
    \dot{y}&=\mu_1+\mu_2 y+x^2\pm xy,
    \end{split}
\end{equation}
where $\mu_1, \mu_2$ are the two unfolding parameters. This form is also manifestly \emph{time-asymmetric}.
\end{theorem}
We shall  split the proof of this theorem  in several steps in the following subsections.

\subsection{Preliminaries}
We shall construct a new basis for the vector space $\mathbb{R}^2$ to transform (\ref{deg-sys}) into a more convenient form. We introduce  two linearly independent vectors,   a $0$-eigenvector $q_0$ and a $0$-\emph{generalized} eigenvector $q_1$ of the matrix,
\be\label{A}
A=
\left(\begin{array}{cc}
  0 & 1 \\
  0 & 0
\end{array}
\right),
\ee
defining the linear part of the vector field  (\ref{deg-sys}), namely,
\be
Aq_0=0,\quad Aq_1=q_0,
\ee
as well as  two adjoint eigenvectors $p_0,p_1$  which belong to the kernel and range of the transposed matrix $A^T$  i.e.,
\be
Ap_1=0,\quad Ap_0=p_1.
\ee
Using the Fredholm altenative, we have the following orthogonality relations\footnote{This is with respect to the standard scalar product of $\mathbb{R}^2$.},
\be
q_0\cdot p_0=1,\quad q_1\cdot p_1=1,
\ee
and
\be
q_1\cdot p_0=0,\quad q_0\cdot p_1=0.
\ee
Solving for the vector components, we can choose four such vectors, for example as follows,
\be
q_0=(1,0),\quad q_1=(1,1),\quad p_0=(1,-1),\quad p_1=(0,1).
\ee
\subsection{Transformation of the system (\ref{deg-sys})}
Any solution of the system (\ref{deg-sys}) (in fact, any plane vector) can now be written as,
\be
(H,\rho)=x q_0 +y q_1,
\ee
where,
\be
x=p_0\cdot(H,\rho),\quad y=p_1\cdot(H,\rho),
\ee
that is,
\begin{equation}\label{t-1}
    \begin{split}
    H&=x+y,\\
    \rho&=y,
    \end{split}
\end{equation}
with inverse, $x=H-\rho,y=\rho$.

Then the FL-system (\ref{deg-sys}) in the new coordinates $x,y$ becomes (after redefining the time as in Eq. (\ref{red-t})),
\begin{equation}\label{deg-sys1a-1}
  \left(
  \begin{array}{c}
    x' \\
  y' \\
  \end{array}
\right)
=
\left(
  \begin{array}{cc}
    0 & 1\\
    0 & 0 \\
  \end{array}
\right)
\left(
  \begin{array}{c}
    x \\
    y \\
  \end{array}
\right)
+
\left(
  \begin{array}{c}
    \Gamma_1 x^2+\Gamma_2 xy+\Gamma_3y^2\\
D_1xy +D_1y^2 \\
  \end{array}
\right),
\end{equation}
where the coefficients are as follows,
\be
\Gamma_1=\frac{6}{3\gamma-2},\quad\Gamma_2=-6,\quad \Gamma_3=-6\frac{3\gamma-1}{3\gamma-2},\quad D_1=\frac{18\gamma}{3\gamma-2}.
\ee

\subsection{Elimination of the nonresonant terms}\label{elim}
We shall use some standard terminology from normal form theory. We denote by $H_2$ the space of all vector-valued homogeneous polynomials of degree two in  $\mathbb{R}^2$. This is a vector space which in the standard basis of $\mathbb{R}^2$ is spanned by the vectors,
\be
H_2=\textrm{span}\left\{\left(
  \begin{array}{c}
    x^2 \\
    0 \\
  \end{array}
\right),\left(
          \begin{array}{c}
            xy \\
            0 \\
          \end{array}
        \right),\left(
          \begin{array}{c}
             y^2 \\
            0 \\
          \end{array}
        \right),\left(
          \begin{array}{c}
            0 \\
            x^2 \\
          \end{array}
        \right),\left(
          \begin{array}{c}
            0 \\
            xy \\
          \end{array}
        \right),\left(
          \begin{array}{c}
            0 \\
            y^2 \\
          \end{array}
        \right)
   \right\},
\ee
obtained by multiplying each of the standard basis vectors of $\mathbb{R}^2$ by the possible monomials of degree two. The space $H_2$ is the direct sum,
\be\label{dir-sum}
H_2=L_J^{(2)}(H_2)\oplus G_2.
\ee
We explain the two terms on the right-hand-side as follows.
For the linear part $A$ given by  Eq. (\ref{A}), it is a simple calculation to show that the first summand is given by,
\be\label{L-space}
L_J^{(2)}(H_2)=\textrm{span}\left\{
\left(
  \begin{array}{c}
    -2xy \\
    0 \\
  \end{array}
\right),\left(
          \begin{array}{c}
           y^2 \\
            0 \\
          \end{array}
        \right),\left(
          \begin{array}{c}
            x^2 \\
            -2xy \\
          \end{array}
        \right),\left(
          \begin{array}{c}
          xy \\
            -y^2 \\
          \end{array}
        \right)
   \right\}.
\ee
This term in the decomposition (\ref{dir-sum}) contains all the `nonresonant terms',  those second-order terms  linear combinations of the basis vectors in Eq. (\ref{L-space})  which can be completely eliminated from the equation (\ref{deg-sys1a-1}).

The last term in (\ref{dir-sum}) is the space complementary to $L_J^{(2)}(H_2)$ that contains all the other terms $F_2^r(Y)$ (`$r$ stands for `resonant') that cannot be in  the range of $L_J^{(2)}(H_2)$ and hence cannot be removed from the equation (\ref{deg-sys1a-1}). That is all other terms can be eliminated except the resonant terms  $F_2^r(Y)$.

It is straightforward to check that the second-order terms in the Eq. (\ref{deg-sys1a-1}) can be further simplified. The $\Gamma$-terms can be rewritten using the $L_J^{(2)}(H_2)$-basis vectors as follows:
\beq
    \Gamma_1\left(
  \begin{array}{c}
x^2\\
0\\
  \end{array}
\right)
+
\Gamma_2 \left(
  \begin{array}{c}
    xy\\
0 \\
  \end{array}
\right)
+
  \Gamma_3 \left(
  \begin{array}{c}
  y^2\\
0 \\
  \end{array}
\right)
&&=
\Gamma_1\left(
  \begin{array}{c}
 x^2 \\
            0 \\
  \end{array}
\right)
+
\Gamma_2 \left(
  \begin{array}{c}
     xy \\
            0 \\
  \end{array}
\right)
+
  \Gamma_3 \left(
  \begin{array}{c}
  y^2\\
0 \\
  \end{array}
\right) \nonumber\\
&&+ \Gamma_1\left(
  \begin{array}{c}
0 \\
            -2xy \\
  \end{array}
\right)
+
\Gamma_2 \left(
  \begin{array}{c}
     0 \\
            -y^2 \\
  \end{array}
\right)\nonumber\\
&&+\Gamma_1\left(
  \begin{array}{c}
0 \\
            2xy \\
  \end{array}
\right)
+
\Gamma_2 \left(
  \begin{array}{c}
     0 \\
     y^2 \\
  \end{array}
\right).\label{Gamma}
\eeq
The first five terms on the right-hand-side are linear combinations of three of the basis vectors in $L_J^{(2)}(H_2)$ and can be eliminated. Therefore substituting the remaining terms present in right-hand-side of Eq. (\ref{Gamma}) in Eq. (\ref{deg-sys1a-1}) in the place of the $\Gamma$-terms, (\ref{deg-sys1a-1}) becomes,
\begin{equation}\label{deg-sys1ab}
  \left(
  \begin{array}{c}
    x' \\
  y' \\
  \end{array}
\right)
=
\left(
  \begin{array}{cc}
    0 & 1\\
    0 & 0 \\
  \end{array}
\right)
\left(
  \begin{array}{c}
    x \\
    y \\
  \end{array}
\right)
+
\left(
  \begin{array}{c}
    0\\
Bxy +Cy^2 \\
  \end{array}
\right),
\end{equation}
where $B=D_1+2\Gamma_1$, and $C=D_1+\Gamma_2$, which is the result announced in Theorem \ref{deg-thm}, Case A.

\subsection{Embedding in a transverse family its final form}
\subsubsection{Step-1}
We now need to embed (\ref{deg-sys1ab}) in a codimension-2 family which is transverse to the versal extension of the linear part given the nilpotent matrix $A$. This is a well-known process and the result is (cf. e.g.,  \cite{wig}, pp. 413-4),
\begin{equation}\label{va-tot1}
    \begin{split}
    \dot{x}&=y,\\
    \dot{y}&=\mu_1 x+\mu_2 y+Bx^2+C xy.
    \end{split}
\end{equation}
Here the $\mu_1,\mu_2$ are the unfolding parameters, and together they constitute the versal unfolding of the linear (nilpotent) part of (\ref{deg-sys1ab}).

\subsubsection{Step-2}
We may further reparametrize to obtain the versal unfolding in its final form, having the origin moving upon parameter variation (this does not occur in (\ref{va-tot1})\footnote{This leads to the notion of fixed \emph{branches} instead of fixed points, cf. \cite{cot24a}.}).
Shifting the $x$-part of the vectorfield, $x\to\bar{x}=x-x_0,y\to y$, the versal family (\ref{va-tot1}) acquires a new constant term,
\begin{equation}\label{va-tot2}
    \begin{split}
    \dot{\bar{x}}&=y,\\
    \dot{y}&=Bx_0^2-\mu_1 x_0 +(\mu_1-2Bx_0)\bar{x}+(\mu_2 -Cx_0)y+B\bar{x}^2+C \bar{x}y.
    \end{split}
\end{equation}

\subsubsection{Step-3}
A further parameter-dependent  shift of the $y$-part,
\be
\bar{y}=y+\frac{\mu_1-Bx_0}{C},
\ee
removes the linear $\bar{x}$-term, and we arrive at the result
\begin{equation}\label{va-tot3}
    \begin{split}
    \dot{\bar{x}}&=\bar{y},\\
    \dot{\bar{y}}&=\bar{\mu}_1 +\bar{\mu}_2 \bar{y}+B\bar{x}^2+C \bar{x}\bar{y},
    \end{split}
\end{equation}
where,
\beq
\bar{\mu}_1&=&Bx_0^2-\mu_1 x_0 +(\mu_1-2Cx_0)(\mu_1-2Bx_0)\\
\bar{\mu}_2&=& \mu_1-2Cx_0.
\eeq

\subsubsection{Step-4}\label{rescale}
Finally, by a standard rescaling in $\bar{x},\bar{y},T$ (cf. \cite{wig}, p. 437, \cite{gh83}, Sect. 7.3), we may reduce the possible cases  to two, namely, in the rescaled variables $B=1$ and $C=\pm1$, which transforms (\ref{va-tot3}) into the result \ref{deg-thm-cusp}, and we have reverted back to unbarred symbols for simplicity without a danger of confusion.

The meaning of the word `cusp' in the title of this Section will be elucidated in Section \ref{ant-sect} below.

\section{The saddle-focus-elliptic codimension-3 FL cubic}\label{3-cub}
In this Section, we are concerned with the system (\ref{deg-sys1}) which contains all the main difficulties associated with, and stemming from, the degeneracies present in  the FL equations.
These difficulties arise and become very transparent and prominent already in the original form of the FL equations (\ref{ds1})-(\ref{ds3}), and are closely related to the presence of the terms $H^2$ and $3\gamma H\rho$ in those equations.

We shall proceed to prove Theorem \ref{deg-thm}, Case B, as well as to develop the following result which settles Case B of Theorem \ref{deg-thm1}.
\begin{theorem}\label{deg-thm-sfe}
When  $\gamma\neq\pm 2/3$, the  FL equations in the form (\ref{deg-sys1b}) (cf. subsection \ref{t-symm case}) versally unfold to
the topological normal form  described by the following versal unfolding,
\begin{equation}\label{va-tot4}
    \begin{split}
    \dot{x}&=y,\\
    \dot{y}&=\mu_1+\mu_2 x+\mu_3y +A x^3 +B xy,
    \end{split}
\end{equation}
where $\mu_1, \mu_2, \mu_3$ are the three unfolding parameters, and
\be\label{a-b}
A=-\frac{108\gamma}{3\gamma-2},\quad B=\frac{6(3\gamma+2)}{3\gamma-2},
\ee
as before.  This form is also manifestly \emph{time-symmetric}.
\end{theorem}
We shall  split our discussion of his result in several steps in the following subsections.

\subsection{The main difficulty}\label{main-dif}
One way to express the main difficulty associated with the second-order terms present in the FL equations (\ref{deg-sys1}) is  to use the decomposition (\ref{dir-sum}). While the term $(0, bH\rho)^T$ in (\ref{deg-sys1}) belongs to the resonant component $G_2$ of the decomposition  (\ref{dir-sum}), the term $(H^2, 0)^T$ belongs only to $H_2$. This simple observation about the structure of the FL equations reveals that the system  (\ref{deg-sys1}) is \emph{not} in normal form.

One can rewrite the second-order term in Eq.  (\ref{deg-sys1}) in the following form,
\begin{equation}\label{deg-term}
  \left(
  \begin{array}{c}
    aH^2 \\
  bH\rho \\
  \end{array}
\right)
=\frac{6}{3\gamma-2}
\left[(3\gamma+2)
  \left(\begin{array}{cc}
    0\\
    H\rho \\
  \end{array}\right)
+
 \left(\begin{array}{cc}
    H^2\\
   -2 H\rho \\
  \end{array}\right)
\right].
\end{equation}
Since the last term is in the $L_J^{(2)}(H_2)$ part of the decomposition (\ref{dir-sum}), it can be completely eliminated and therefore the FL equation  (\ref{deg-sys1}) is replaced by,
\begin{equation}\label{deg-sys1a-2}
  \left(
  \begin{array}{c}
    H' \\
  \rho' \\
  \end{array}
\right)
=
\left(
  \begin{array}{cc}
    0 & 1\\
    0 & 0 \\
  \end{array}
\right)
\left(
  \begin{array}{c}
    H \\
    \rho \\
  \end{array}
\right)
+
B\left(
  \begin{array}{c}
    0\\
H\rho  \\
  \end{array}
\right),
\end{equation}
where $B$ is as in Eq. (\ref{deg-sys1a1}).

There are two features of  the system (\ref{deg-sys1a-2}) which will play an important role in the following. First, (\ref{deg-sys1a-2}), like the previous form Eq. (\ref{deg-sys1}), is manifestly time-symmetric in the sense of being $(\mathbb{Z}_2 +t)$-equivariant that is invariant under the symmetry (\ref{t-symm}). However, Eq. (\ref{deg-sys1a-2}) is qualitatively \emph{inequivalent} to the Bogdanov-Takens normal form found in Section \ref{cusp}, namely Eq. (\ref{va-tot3}), because the $x^2$ term  is missing presently. This means that the present case is more degenerate than that of the Section \ref{cusp}, in particular it is expected to contain third-order terms.

\subsection{The normal form}
In fact, by a generalization of the standard normal form process leading to second-order terms we may acquire a first taste of the structure of the sought-for normal form containing the appropriate third-order terms. To deal with the above difficulty, we proceed as follows (cf. e.g.,  \cite{perko}, p. 167). In general, one starts from the system (\ref{deg-sys1}) and looks for the normal form to $O(3)$ terms in $X=(H,\rho)^T$, which we write symbolically as $\dot{X}=JX+F_2(X)+F_3(X)+O(|X|^4)$, where $F_i\in H_i,i=2,3$, and $H_i$ is as in Section \ref{elim}, and \cite{cot24a}, Section 7.1.2, respectively, for $i=2,3$. Then by a series of near-identity transformations, we first reduce this system to the form $\dot{X}=JX+\bar{F}_2(X)+\tilde{F}_3(X)+O(|X|^4)$, where $\bar{F}_2(X)\in G_2$ and $\tilde{F}_3(X)\in H_3$, and secondly, by another nonlinear transformation to the final form $\dot{X}=JX+\bar{F}_2(X)+\hat{F}_3(X)+O(|X|^4)$, with $\hat{F}_3(X)\in G_3$.

The first reduction was already performed in the previous subsection, where we arrived at the reduced form (\ref{deg-sys1a-2}). As it turns out, the following smooth change accomplishes the second reduction, namely, under the transformation of the field,
\begin{equation}\label{smoo-cha}
  \left(
  \begin{array}{c}
    H \\
  \rho \\
  \end{array}
\right)
=
  \left(\begin{array}{cc}
    y_1\\
    y_2 \\
  \end{array}\right)
+
 \left(\begin{array}{cc}
    0\\
   -ay_1^2 \\
  \end{array}\right),
\end{equation}
Eq. (\ref{deg-sys1}) finally becomes the FL `cubic'\footnote{We shall use this word to indicate the situation that the FL equations when reduced to a normal form contain a cubic term.},
\begin{equation}\label{nor-form3}
    \begin{split}
    \dot{x}&=y,\\
    \dot{y}&=A x^3 +B xy,
    \end{split}
\end{equation}
so that the Theorem \ref{deg-thm}, Case B follows (for simplicity we have reverted back to the variables $x,y$ in the place of the transformed variables $y_1,y_2$). The important new feature in the reduction (\ref{nor-form3}) is that the third-order vector field $(0,Ax^3)^T$ belongs to $G_3$. Thus we have succeeded in finding the normal form of the FL equations in this case.

Finally, we note that using a similar rescaling as that which was used in Section \ref{rescale}, namely, $x\to\alpha x,y\to\beta y, t\to \delta t, \delta>0$, without loss of generality we may set $A=\pm 1, B>0$ in Eq. (\ref{nor-form3}).

\subsection{Andronov theorem}\label{ant-sect}
Before we proceed further, we make the following observation. In the analysis of linearized stability summarized in Section \ref{lin-stab}, we found that the two equilibria,  de Sitter space and the Einstein static universe, were saddles and sinks respectively, and this appeared as a basic feature of the FL equations.

However, this conclusion appears perhaps somewhat restricted for a number of reasons which will unravel gradually in the next few Sections of this paper. Using two of the results we have developed up to now, namely,  Eqns. (\ref{deg-sys1ab}), (\ref{nor-form3}), we see that the FL equations contain \emph{non-hyperbolic points} instead of the assumed saddle or a sink in linearized studies. This has important consequences because when $\gamma\neq 2/3$ the behaviour of the solutions near both the de Sitter and the ESU depends on whether we assume a time-asymmetric  or a time-symmetric system leading to the equations (\ref{deg-sys1ab}), (\ref{nor-form3}), hence giving a cusp or a elliptic domain respectively,   not a saddle or a sink.

For planar systems with non-hyperbolic points of nilpotent type, there is a general theorem discovered by Andronov and coworkers which characterizes completely all possible dynamical behaviours near such points (cf.  \cite{ant1}, Section 22 for a proof). We shall state Andronov's theorem  in its simplest form in order to give a broad classification of the possible behaviours of solutions to the FL equations (\ref{deg-sys1ab}), (\ref{nor-form3}), before we proceed to the analysis of their bifurcations.

\begin{theorem}[Andronov \emph{et al} 1973]  Part 1: General normal form.

The normal form of a planar dynamical system with nilpotent linear part is given by,
\begin{equation}\label{ant-thm}
    \begin{split}
    \dot{x}&=y,\\
    \dot{y}&=a_kx^k (1+h(x))+b_nx^ny (1+g(x))+y^2 f(x,y),
    \end{split}
\end{equation}
where $h(x),g(x),f(x,y)$ are analytic in a neighborhood of the origin, $h(0)=g(0)=0,$ $k\geq 2,a_k\neq 0$, and $n\geq 1, b_n\neq 0$.
\end{theorem}

We shall list below all possible non-hyperbolic points only for the simplest, that is lowest-order, values of the exponent $k$. When $k$ is even, $k=2m$, we shall list only the cases with $m=1$, and similarly  when $k$ is odd, $k=2m+1$.

\begin{theorem}[Andronov \emph{et al} 1973] \  Part 2: Lowest exponents.

There are two main quadratic cases, and five cubic ones in total. We have:
\begin{itemize}
    \item\textsc{Quadratic cases, $k=2$.}
\begin{itemize}
\item \underline{Case $A_1$: Cusp.} In this case, $a_{2m}>0$.
 \begin{enumerate}
    \item $b_n=0$.

    Simplest case:  We have the form: \framebox{$a_2x^2$}.
    \item $b_n\neq 0.$

    Simplest case: $n=m=1$, $b_1\neq 0$. We have the form: \framebox{$a_2x^2+b_1 xy$}. This case leads to a codimension-2 Bogdanov-Takens 2nd-order.
 \end{enumerate}
\item \underline{Case $A_2$: Saddle node}: $b_n\neq 0.$ In this case, $a_{2m}<0$.

Simplest case: $n=0<1=m$. We have the form: $b_0 y+a_2x^2.$
\end{itemize}
    \item\textsc{Cubic cases, $k=3$.}

    \begin{itemize}
\item \underline{Case $B_1:$ \framebox{$a_3>0$}: Topological saddle}
\item \underline{Case $B_2: a_3<0$}: Four main cases as per below. Behaviours depend on the expression $\lambda:=b_n^2+4(m+1)a_3.$ We note the two further subcases when $b_n>0$ and $b_n<0$
\begin{enumerate}
\item \underline{a. Focus, or, b. centre}
\begin{enumerate}
\item $b_n= 0$. (No $xy$ term with $n=1$.)
\item $b_n\neq 0$ and $n>m$.

When $m=1$, the simplest case is $n=2$, so we get \framebox{$a_3x^3+b_2x^2y$}: Codimension-2 Bogdanov-Takens 3rd-order.
\item When $n=m$  and $\lambda<0$, simplest case $n=m=1$, we find \framebox{$a_3x^3+b_1xy$}: Codimension-3 versal unfolding.
\end{enumerate}
\item \underline{Topological node}
\begin{enumerate}
\item  $b_n\neq 0$ and $n$ even $<m$, simplest case $m=3,n=2$, we find \framebox{$a_7x^7+b_2x^2y$}.
\item $b_n\neq 0$ and $n=m$ even, so simplest case is $m=2$, and we get \framebox{$a_5x^5+b_2x^2y$}, with $\lambda=b_2^2+12 a_5\geq 0$.
\end{enumerate}
\item \underline{Elliptic point}

We note the two further subcases when $b_n>0$ and $b_n<0$, phase portrait is upside down (reflection in the $x$-axis) when $b_n<0$.
\begin{enumerate}
\item $b_n\neq 0$ and $n$ odd $<m$, simplest case $m=2,n=1$, we find \framebox{$a_5x^5+b_1xy$}.
\item  $b_n\neq 0$ and $n$ even $=m$, simplest case $n=m=1$, we find \\ \framebox{$a_3x^3+b_1xy$}, and so this is a codimension-3 versal unfolding.
\end{enumerate}
\end{enumerate}
\end{itemize}
\end{itemize}
\end{theorem}\label{ant-thm2}

Andronov's theorems give the normal forms of systems with  linear part given by (\ref{A}). We see that instead of ending up with the linear saddles, nodes, or foci met  the linearized studies, the FL equations admit a great variety of nonlinear, non-hyperbolic points: besides  topological (nonlinear) versions of the saddle, nodes or foci, they admit in addition, the cusps, the saddle-nodes, and the elliptic domain points. We shall see below that in fact these possibilities correspond in turn to only the simplest behaviours to be met in the FL equations. This is  because the full  description of the FL dynamics requires by necessity  taking into account their \emph{bifurcations}.

As an immediate application of  Theorem \ref{ant-thm2} to our problem, we get the following result which justifies the name `FL-cusp' in Section \ref{cusp}.
\begin{corollary}
At the origin of the parameter space $\mu_1,\mu_2=0$ in Theorem \ref{deg-thm-cusp}, the FL dynamics under the assumptions of that theorem is described by the \textbf{cusp} found Theorem \ref{ant-thm2}, Case $A_1$-2.
\end{corollary}

\subsection{The versal unfolding of the codimension-3  FL-cubic}
We now return to the FL system (\ref{deg-sys1a-2}) and consider the question of what is the topological normal form, or versal unfolding, in this case. We know from Section \ref{main-dif} that the FL equations in this case are more degenerate than their time-asymmetric counterpart Eq. (\ref{deg-thm-cusp}) analysed in Section \ref{cusp}, and we also learnt from the previous subsection that the normal form necessarily contains cubic terms.

In fact, this problem was only solved somewhat recently, in the fundamental work \cite{du1,du2} where it is shown that the topological normal form of Theorem \ref{deg-thm-sfe} is equivalent to the following versal unfolding,
\begin{equation}\label{va-tot5}
    \begin{split}
    \dot{x}&=y,\\
    \dot{y}&=\mu_1+\mu_2 x+\mu_3y \pm x^3 +B xy,
    \end{split}
\end{equation}
and we only need to consider the case $B>0$ (the other case is found by a reflection about the $x$-axis (cf. Theorem \ref{ant-thm2}). Therefore for the current cubic case we arrive at the following result.
\begin{corollary}\label{cor-2}
At the origin of the parameter space $\mu_1,\mu_2,\mu_3=0$ in Eq.  (\ref{va-tot5}), the FL dynamics is described by the following cases:,
\begin{itemize}
\item $(+)$-sign case in (\ref{va-tot5}) corresponds to a \textbf{topological saddle}  (Case  $B_1$ of Theorem \ref{ant-thm2}).
\item $(-)$-sign case in (\ref{va-tot5}) corresponds to either a \textbf{focus} or an \textbf{elliptic domain}   (cases  $B_2 -1c$ and $B_2 -3b$ of Theorem \ref{ant-thm2} respectively.)
\end{itemize}
\end{corollary}
Since by Eq. (\ref{a-b}), $A\gtrless 0$ iff the fluid parameter $\gamma\lessgtr 0$, we conclude by an application of this Corollary that the FL cubic is a topological saddle at the origin for $\gamma<0$ (i.e., the $(+)$-sign case in the Corollary \ref{cor-2}). Using the combination $\lambda=B^2+8A$, we also conclude the the origin is a focus when $B<2\sqrt{2}$, and an elliptic domain when $B>2\sqrt{2}$. Since the form (\ref{va-tot5}) is manifestly time-symmetric our discussion of Theorem \ref{deg-thm-sfe} in the beginning of this Section is thus concluded.

\section{The cases  $\gamma=-2/3$ and $\gamma=2/3$}\label{gam}
In this Section, we complete our discussion of the possible versal unfoldings of the FL equations by treating the remaining two, single $\gamma$-values, cases.
\subsection{$\gamma=-2/3$: the saddle-focus codimension-2 FL cubic}
Let us first treat the Case C in Theorem \ref{deg-thm}, namely, the case $\gamma=-2/3$ for the fluid parameter $\gamma$. Using the method of Section \ref{elim}, we conclude the the quadratic term in the FL equation (\ref{deg-sys2}) of Section \ref{-2/3}, namely, the term
\be
-\frac{3}{2}\left(
  \begin{array}{c}
    H^2\\
-2H\rho \\
  \end{array}
\right),
\ee
is a linear combination of nonresonant vector fields in the space  $L_J^{(2)}(H_2)$ (cf. Eq. (\ref{L-space})) and can thus be eliminated. Hence we must proceed to calculate third-order terms, that is the space $L_J^{(3)}(H_3)$.

This is a problem identical to the one addressed in \cite{cot24a}, Sections 7.1.2, 7.1.3, and we shall present a very brief discussion of the solution leading to the versal unfoldings in this case.

The space $H_3$ of all vector-valued homogeneous polynomials of degree three in  $\mathbb{R}^2$ is given by,
\beq
H_3&=\textrm{span}\left\{
\left(
  \begin{array}{c}
    x^3 \\
    0 \\
  \end{array}
\right),
\left(
\begin{array}{c}
            x^2y \\
            0 \\
          \end{array}
        \right),
        \left(
          \begin{array}{c}
             xy^2 \\
            0 \\
          \end{array}
        \right),
        \left(
          \begin{array}{c}
             y^3 \\
            0 \\
          \end{array}
        \right),
        \left(
          \begin{array}{c}
            0 \\
            x^3 \\
          \end{array}
        \right)\nonumber\right.,\\ &\left.
        \left(
          \begin{array}{c}
            0 \\
            x^2y \\
          \end{array}
        \right),
        \left(
          \begin{array}{c}
            0 \\
            xy^2 \\
          \end{array}
        \right),
        \left(
          \begin{array}{c}
            0 \\
            y^3 \\
          \end{array}
        \right)
   \right\}
\eeq
As in the second-order terms case, this space is the direct sum,
\be\label{dir-sum2}
H_3=L_J^{(3)}(H_3)\oplus G_3,
\ee
where,
\beq
L_J^{(3)}(H_3)&=\textrm{span}\left\{\left(
  \begin{array}{c}
    x^3\\
    -3x^2y \\
  \end{array}
\right),\left(
  \begin{array}{c}
    x^2 y\\
    0 \\
  \end{array}
\right), \left(
  \begin{array}{c}
    xy^2 \\
    0 \\
  \end{array}
\right),\nonumber\right.\\ &\left.\left(
  \begin{array}{c}
    y^3\\
    0 \\
  \end{array}
\right),\left(
  \begin{array}{c}
    0 \\
    y^3 \\
  \end{array}
\right),\left(
  \begin{array}{c}
    0 \\
    xy^2 \\
  \end{array}
\right)
  \right\},
\eeq
while the resonant component is found to be,
\begin{equation}
G_3=\textrm{span}
\left\{
\left(
  \begin{array}{c}
    0 \\
    x^3 \\
  \end{array}
\right),
\left(
  \begin{array}{c}
    0\\
    x^2y \\
  \end{array}
\right)
\right\}.
\end{equation}
Following the discussion in \cite{cot24a}, we find that after rescaling the normal form of the FL equation (\ref{deg-sys2}) is given by,
\begin{equation}\label{nor-form-3rd}
\begin{split}
\dot{x}&=y,\\
\dot{y}&=\pm x^3-x^2y.
\end{split}
\end{equation}
We therefore arrive at the following result for the versal unfolding of the FL equations in this case.
\begin{theorem}\label{deg-thm-3rd}
When  $\gamma=-2/3$, the  FL equations in the form (\ref{deg-sys2a}) versally unfold to
the topological normal form  described by the third-order Bogdanov-Takens versal unfolding, namely,
\begin{equation}\label{va-3rd}
\begin{split}
\dot{x}&=y,\\
\dot{y}&=\mu_1 x+\mu_2 y\pm x^3-x^2y,
\end{split}
\end{equation}
where $\mu_1, \mu_2$ are the two unfolding parameters. This form is also manifestly \emph{time-symmetric}.
\end{theorem}

\subsection{$\gamma=2/3$: the $\mathbb{Z}_2$-equivariant codimension-2 unfolding}
We now move on to discuss the last case of Theorem \ref{deg-thm}, case D, namely the system (\ref{deg-sys3}) which we reproduce here (also setting $x:=H,y:=\rho$ for uniformity of notation, and rescaling $t\to -t$),
\be \label{deg-sys3-1}
\left(
  \begin{array}{c}
    \dot{x} \\
    \dot{y} \\
  \end{array}
\right)
=\left(
  \begin{array}{c}
    x^2\\
2 xy \\
  \end{array}
\right).
\ee
As we discussed in Section \ref{2/3}, the linear part of this system has also two zero eigenvalues but instead of being nilpotent like in all previous cases, it is the zero matrix. This system  (\ref{deg-sys3-1})  is also invariant under the $\mathbb{Z}_2$-symmetry $x\to x,y\to -y$.

In the literature such systems are called truncated amplitude systems because the may be regarded as reductions of three-dimensional systems having an extra equation that describes a rotation around the $x$-axis.  In gravitation such systems appear as  governing the dynamics of the optical scalars in spacetime, and also as the convergence-vorticity systems associated with vorticity congruences, cf. \cite{cot24a}, Sections 5, 6. It is interesting that exactly the same situation arises in the totally distinct FL, $\gamma=2/3$ case presently.

The versal unfolding of the Eq. (\ref{deg-sys3-1}) is given by \cite{zol84}\footnote{This versal family is equivalent to the one used in \cite{zol84}. As Zholondek proves in his work, that equation represents all non-degenerate (versal) unfoldings. The remaining unfoldings are degenerate, non-transversal families which are described as a finite union of codimension-1 submanifolds.},
\begin{equation}\label{va-2/3}
\begin{split}
\dot{x}&=\mu_1+x^2+sy^2 ,\\
\dot{y}&=\mu_2 y+ 2xy +x^2y,
\end{split}
\end{equation}
where $\mu_1,\mu_2$ are two unfolding parameters, and $s=\pm 1$ is another modular parameter. The properties of this versal family have been discussed in many references and the bifurcation diagrams have been obtained (cf. \cite{cot24a} and refs. therein).

\section{Bifurcation diagrams and FL metamorphoses}\label{bifn-dia}
In this Section, after discussing the role of the fluid parameter $\gamma$, we  present some aspects of the very rich patterns of bifurcations that occur at different $\gamma$ parameter ranges and which describe the amazing variety of dynamical features that stem from the four versal unfoldings  of the FL universes, namely, the versal families (\ref{va-tot}), (\ref{va-tot4}), (\ref{va-3rd}), and (\ref{va-2/3}). We shall also discuss the `rough' nature of the emerging solutions.

A  very important aspect of the versal solutions is their continuous  `metamorphoses' (or `perestroikas' in other terminology \cite{ar94}) into new forms of evolution effected by some of the bifurcations, e.g., the pitchfork. New solutions are typically  created by  saddle-node bifurcations and then revisited by other kinds of transitions. When this happens one observes an old structure transfiguring to another one precisely when the system passes through a bifurcation set in the corresponding parameter space.

A bifurcation diagram is to a versal unfolding what a solution set is to a hyperbolic dynamical system, namely, it contains all possible qualitative information about the behaviour of the solutions.
For a given versal family, the bifurcation diagrams  consist two parts,
\begin{itemize}
\item The parameter diagram.

This is the parameter space (2-, or 3-dimensional in all cases below)  divided into separate connected regions, the \emph{strata}, at each point of which the system has a phase portrait topologically equivalent to that of any other point in the same region. The strata are separated by well-defined \emph{bifurcation sets}, points, curves or higher-dimensional manifolds in parameter space where all qualitative changes in the dynamics take place.
\item The phase portraits.

These are representative of the dynamics in each one of the strata of the parameter diagram.
\end{itemize}
Normally a bifurcation diagram appears as a superposition of phase portraits on the parameter space (one per stratum separated by the bifurcation sets).

Once the bifurcation diagram(s) of a system is known, we are equipped with a complete classification scheme of all possible modes of behaviour of the solutions of the system as well as their bifurcations. That is, we know the possible transitions between solutions upon parameter variation, when a particular type of universe is transfigured into another and how this metamorphosis is exactly done. Given that the versal unfolding of a given equation contains information about \emph{all} possible perturbations of that equation, the bifurcation diagram is in a sense the `theory of everything' of the equation one started with - it contains the answer to any question that may be posed about that equation.
 \begin{figure}
\centering
\includegraphics[width=\textwidth]{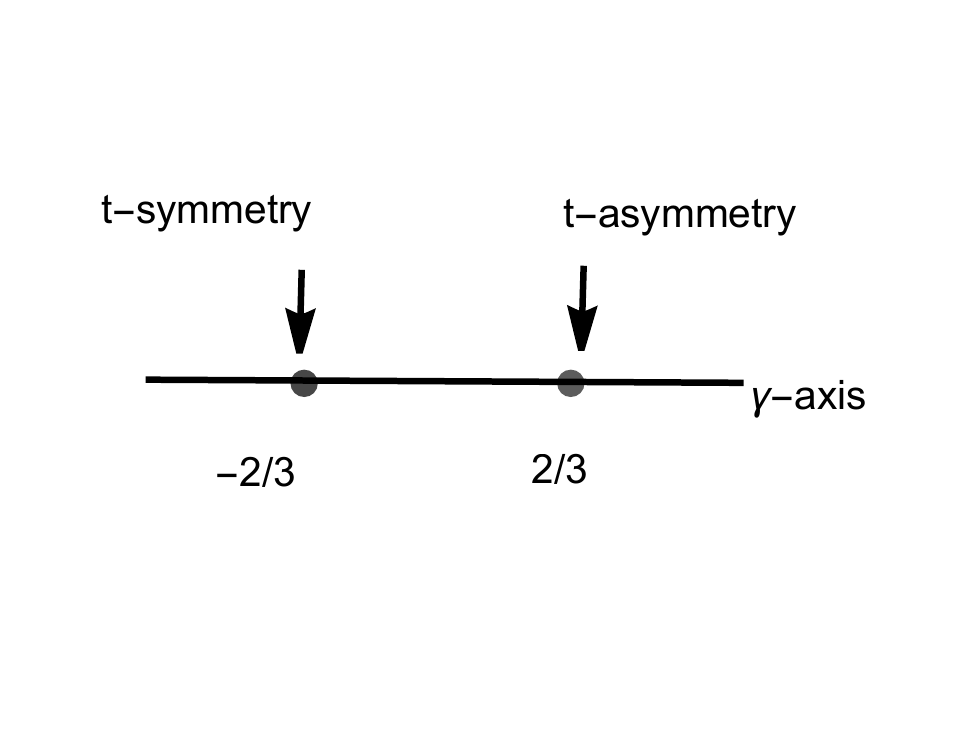}
\caption{The axis of the  parameter $\gamma$. Away from the two shown values of $\gamma$, the system can be either time-symmetric or time-asymmetric. In either case, a passage through the time-(a)symmetric value $-2/3$ (resp. $2/3$) dualizes the behaviour of a time-(a)symmetric that existed for the system when $\gamma\neq\pm 2/3$.}\label{t-gamma}
\end{figure}

\subsection{The role of the fluid parameter $\gamma$}\label{roleg}
All different versal unfoldings of the FL equations emerge in the five possible sets of values of the fluid parameter $\gamma$ and can be classified as to whether or not they  are $(\mathbb{Z}_2 +t)$-equivariant. When they are, we call them time-symmetric, otherwise they are time-asymmetric. This is shown in Fig. \ref{t-gamma}.

To understand what the role of the $\gamma$-parameter is, suppose that the system is time-asymmetric for some value of $\gamma<-2/3$, and we  wish to understand the behaviour of the system as $\gamma$ increases (motion to the right of the $\gamma$-axis in Fig. \ref{t-gamma}). Then on passing through the value $\gamma=-2/3$ the system becomes manifestly time-symmetric and remains so until it crosses the value $\gamma=2/3$ where it changes to time-asymmetry and keeps it for $\gamma>2/3$. The system may be either time-symmetric or asymmetric for any $\gamma\neq \pm 2/3$, but it is the crossing of the two critical values $\pm 2/3$ that changes the type of time-symmetry the system possesses.

In this example, when  the system is in the $\gamma<-2/3$ region of the  $\gamma$-parameter, it behaves like an FL cusp (that is it is time-asymmetric) with bifurcation diagrams described as in Subsection \ref{type-1} below. When $\gamma$ crosses the value $-2/3$, the system dualizes to time-symmetry and the bifurcation diagrams of \ref{type-3}, then in the region $-2/3<\gamma<2/3$ remains time-symmetric with bifurcation diagrams as those described in subsection \ref{type-2}, then when $\gamma=2/3$ the system dualizes again back to time-asymmetry and the bifurcation diagrams are those described in \ref{type-4}, to finally change  to those of subsection \ref{type-2} (retaining the time-asymmetry). All such evolutionary changes are smooth, the universe point moves on  some orbit of a phase portrait in some bifurcation diagram, continues its motion from phase portrait to the next passing through subsequent bifurcation sets, strata, and eventually through whole bifurcation diagrams.

We conclude that the role of $\gamma$ in the evolution of the FL equations is to drive the system between topologically inequivalent bifurcation diagrams available, while those of the unfolding parameters ($\Lambda$ included) are to drive the system through each bifurcation diagram.
Once the system is found itself in one of the five possible `ranges' of $\gamma$-values (the two point values are also regarded as two of these ranges), the consequent behaviour of the system is completely determined by the unfolding parameters $\mu_i,i=1,2$ (and also $\mu_3$ for the codimension-3 unfolding (\ref{va-tot4}))  present in the corresponding versal families and their bifurcation diagrams. (We note again that since in all cases we are unfolding from zero $\lambda$, one of the unfolding parameters must be the cosmological constant $\Lambda$.)

\subsection{Features of the bifurcation diagrams}\label{b-dia}
We describe below some aspects of these diagrams. Overall, one counts nine possible topologically distinct  bifurcation diagrams for the FL equations. We emphasize  that the (normal forms of the) original FL equations are found when ones sets the $\mu$ parameters equal to zero.
\subsubsection{The $\gamma=-2/3$ case}\label{type-3}
The bifurcation diagrams correspond to the cubic Bogdanov-Takens bifurcation (\ref{va-3rd}), namely,
\begin{equation}\label{va-3rd-A}
\begin{split}
\dot{x}&=y,\\
\dot{y}&=\mu_1 x+\mu_2 y\pm x^3-x^2y,
\end{split}
\end{equation}
are shown in Figs. \ref{fig2}, \ref{fig3} for the positive and negative moduli respectively.

\begin{figure}
\centering
\includegraphics[width=\textwidth]{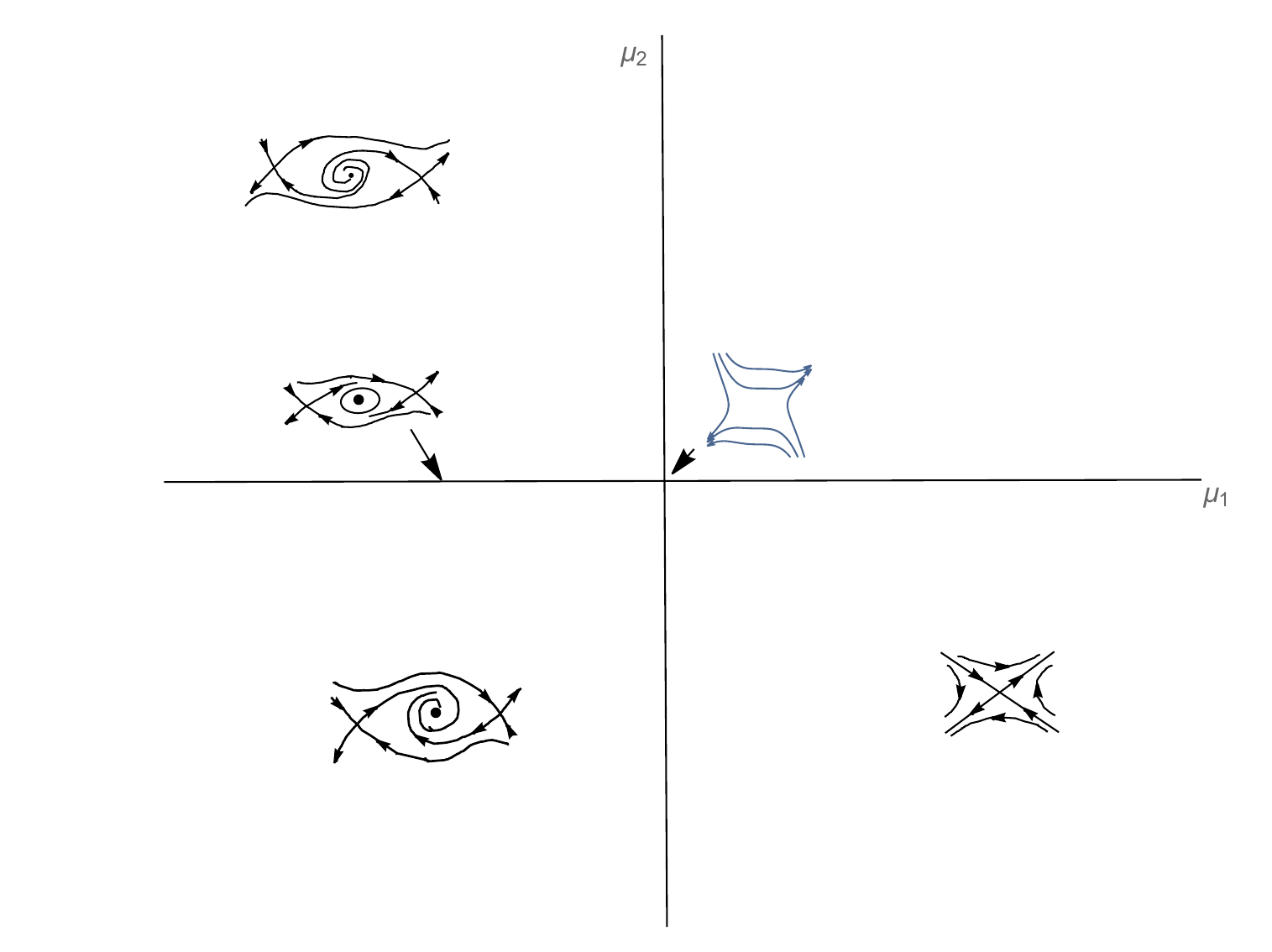}
\caption{The  bifurcation diagram for the $\gamma=-2/3$, positive moduli case. }\label{fig2}
\end{figure}
 \begin{figure}
\centering
\includegraphics[width=\textwidth]{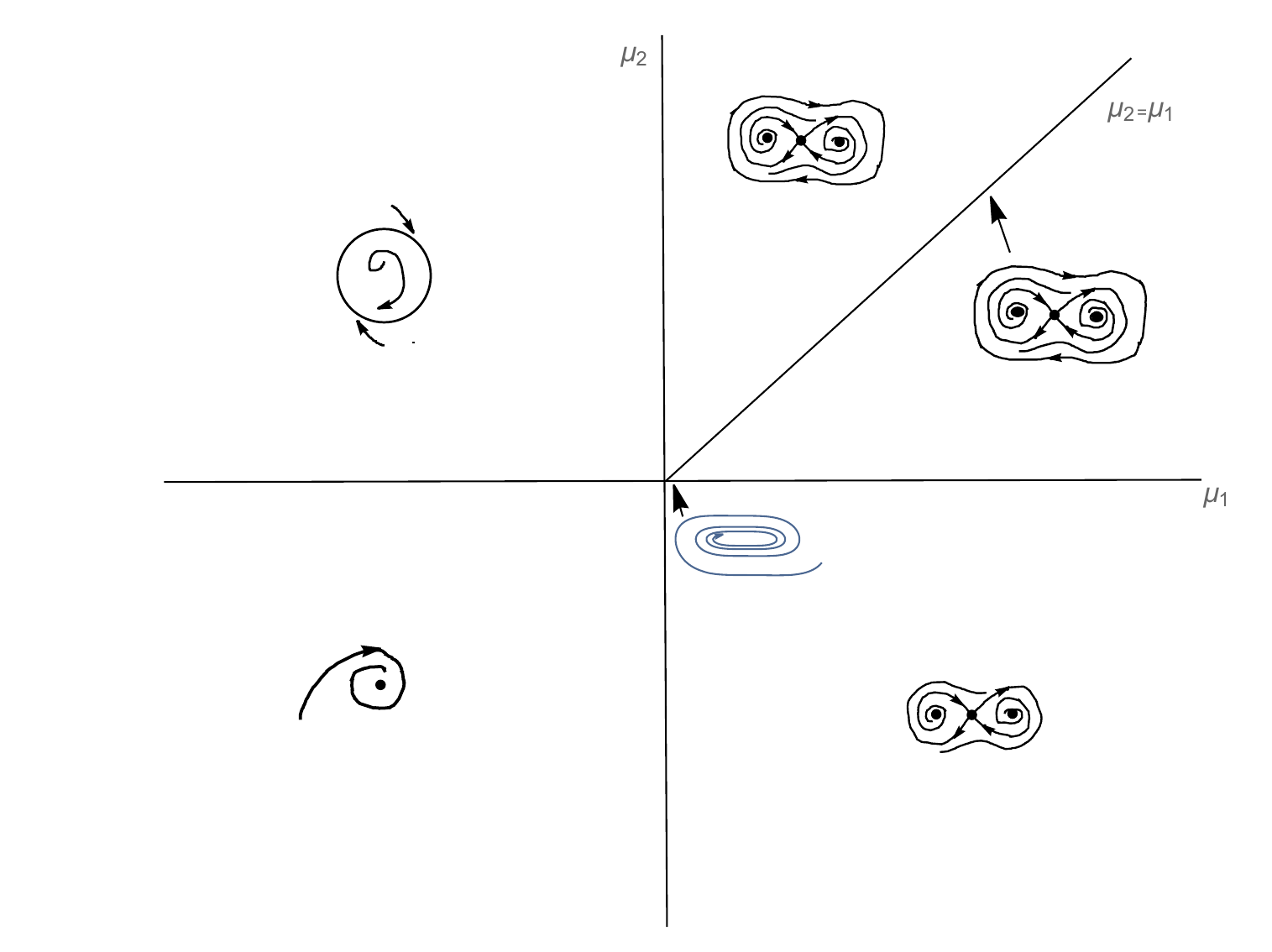}
\caption{The  bifurcation diagram for the $\gamma=-2/3$, negative moduli case. }\label{fig3}
\end{figure}
It is interesting that the same bifurcation dynamics occurs here as in \cite{cot24a}, figs. 13, 14 for the $(+)$ and the  $(-)$ case respectively. We shall only briefly recall here the main features of the $(+)$-case, directing the reader to the Section 7.3 of \cite{cot24a} for the other case.
We summarize in a qualitative way the $(+)$-case following closely \cite{cot24a}, Sect. 7.2, (however, much of the intriguing difficulties in obtaining these results are now lost in the present summary). The origin $\mathcal{E}_{3}=(0,0)$ is a saddle for $\mu_1>0$ and any sign of $\mu_2$, while when $\mu_1<0$, it is  a source for $\mu_2>0$, and a sink for $\mu_2<0$. Apart for the origin, there are two \emph{new} equilibria $\mathcal{E}_{1,2}=(\pm\sqrt{-\mu_1},0)$ which are real provided $\mu_1<0$. They are both saddles.

There are two kinds of local bifurcations:
\begin{itemize}
  \item \emph{A pitchfork bifurcation when $\mu_1=0$.} So in this case the behaviour of the system near $\mu_1$ is supercritical (resp. subscritical) for $\mu_2>0$ and the origin repels all orbits (resp.  $\mu_2<0$ the origin is an attractor to all orbits).
  \item \emph{A Hopf bifurcation when $\mu_2=0$.} When  $\mu_1<0$, this is supercritical (the bifurcating limit cycle is stable at $\mu_2=0$). We note that there are various regions of parameter space with no periodic solutions.
\end{itemize}
In the $(-)$   case, we have again analogous bifurcations but all conclusions are `dualized' (`super'$\to$`sub', etc), while the new equilibria are centers. We also note that the origin in parameter space is now a \emph{focus} (instead of a saddle in the $(+)$-case).

\subsubsection{The $\gamma=2/3$ case}\label{type-4}
The bifurcation diagrams correspond to those of the $\mathbb{Z}_2$-symmetric versal family (\ref{va-2/3}), namely,
\begin{equation}\label{va-2/3-A}
\begin{split}
\dot{x}&=\mu_1+x^2+sy^2 ,\\
\dot{y}&=\mu_2 y+ 2xy +x^2y,
\end{split}
\end{equation}
are shown in Figs. \ref{fig-4}, \ref{fig-5} for  $s=\pm 1$ respectively (interestingly, similar dynamics is found in  \cite{cot24a}, fig. 8 for the $(+)$-case, and fig. 11 for $(-)$-case there). We again only briefly summarize here the main features of the ensuing bifurcations.

For the $(+)$-case (cf. \cite{cot24a}, Sect. 5), at zero parameter the dynamics is characterized by the existence of the straight line separatrices, $y=\pm x$, a key fact for the construction of the phase portrait in this case. All solutions are unstable (and the origin represents only de Sitter space in the present subsection).

\begin{figure}
\centering
\includegraphics[width=\textwidth]{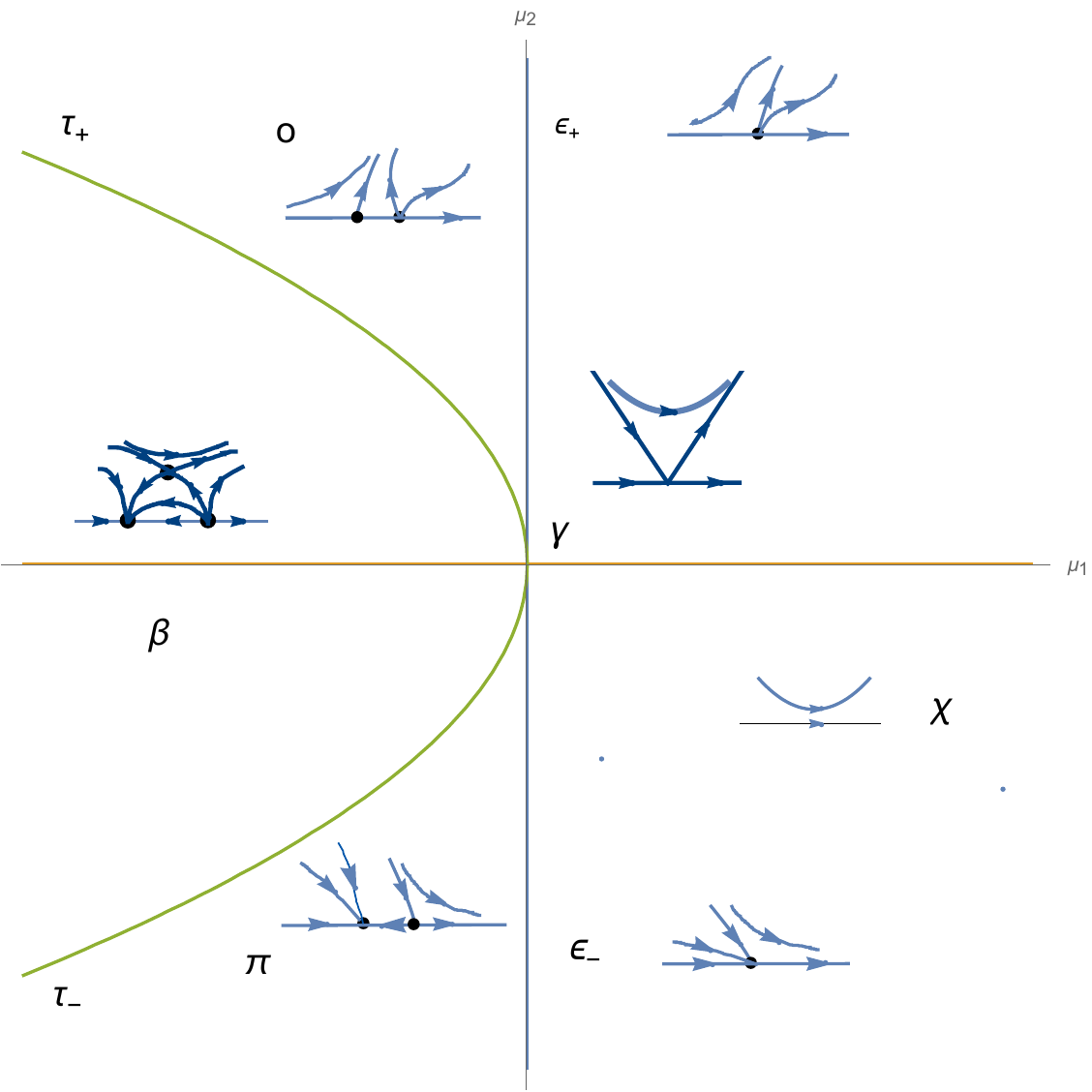}
\caption{The  bifurcation diagram for the $\gamma=2/3$, positive moduli case. }\label{fig-4}
\end{figure}
 \begin{figure}
\centering
\includegraphics[width=\textwidth]{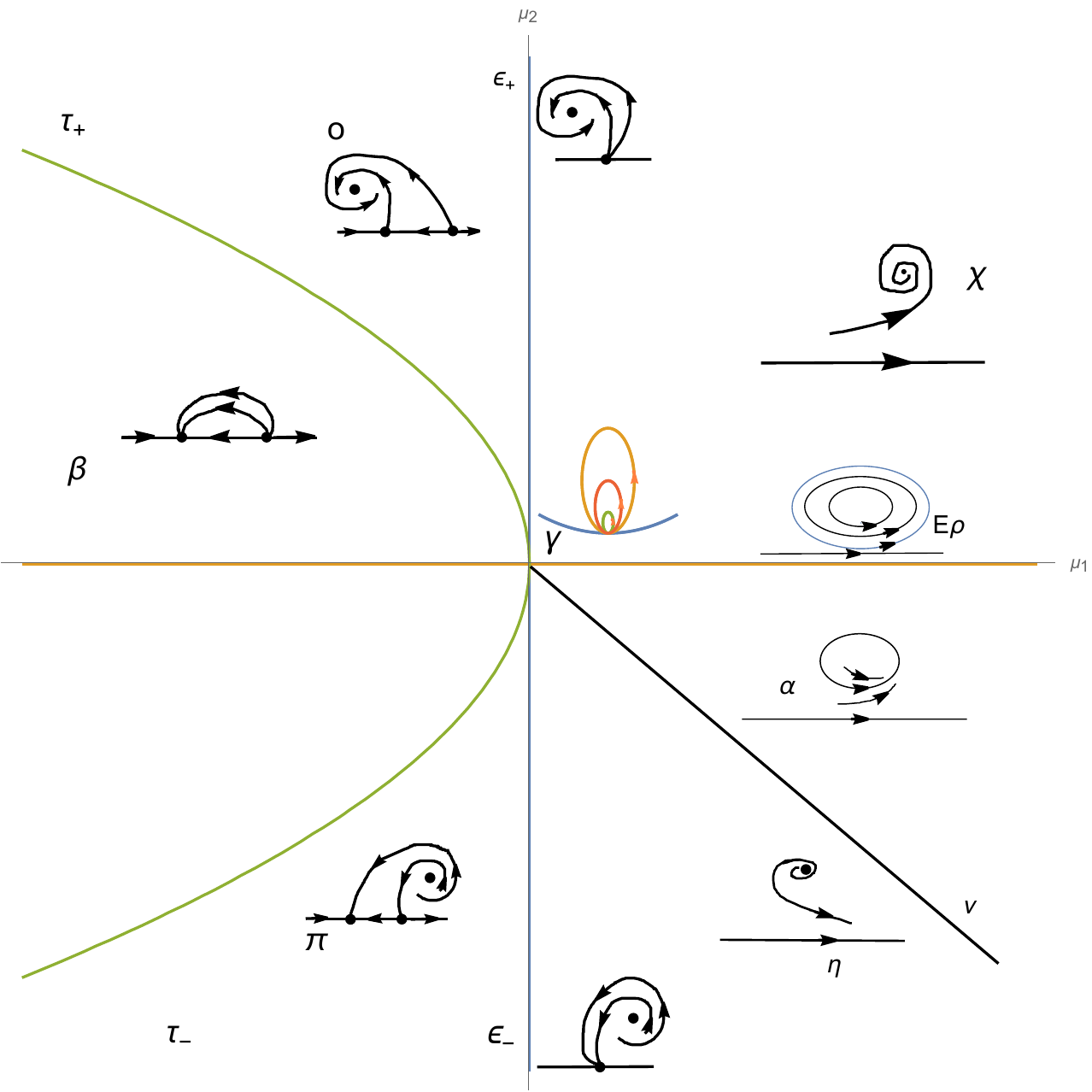}
\caption{The  bifurcation diagram for the $\gamma=2/3$, negative moduli case. }\label{fig-5}
\end{figure}
There are three new equilibria (i.e., `fixed branches') in this case, all parameter-dependent, but none of them exists in the right half of the parameter space $(\mu_1,\mu_2)$:
\begin{enumerate}
  \item The first two are, $\mathcal{E}_{1,2}=(\mp\sqrt{-\mu_1},0)$, $\mu_1<0.$
  \item The third fixed branch is,
\be\label{e3}
\mathcal{E}_{3}=\left(-\frac{\mu_2}{n},\sqrt{-\left(\frac{\mu_2^2}{n^2}+\mu_1 \right)}\right),\quad\mu_1<-\mu_2^2/n^2.
\ee
\end{enumerate}
The stability of the equilibria $\mathcal{E}_{1,2}$ depends on the sign of the combination $\mu_2-2\sqrt{-\mu_1}$, $\mathcal{E}_{1}$ is a saddle or sink, and $\mathcal{E}_{2}$ a source or saddle, when $\mu_2-2\sqrt{-\mu_1}\gtrless 0$ in the two case respectively, while  $\mathcal{E}_{3}$ is a saddle.

The equilibria $\mathcal{E}_{1,2}$ are created in a saddle-node bifurcation, in the sense suggested by the centre manifold dynamics, $\dot{x}=\mu_1+x^2,$ which also implies that the saddle-node is solely $H$- (i.e., $x$)-dominated. On the other hand, the equilibrium $\mathcal{E}_{3}$ is \emph{created} in a supercritical pitchfork bifurcation by the equilibrium $\mathcal{E}_{1}$ crossing the $\tau_+$-branch: $\mu_1=-\mu_2^2/4, \mu_2>0$ transversally, and also in a subcritical pitchfork bifurcation by the transversal crossing of the  $\tau_-$-branch: $\mu_1=-\mu_2^2/4, \mu_2<0$  by the equilibrium $\mathcal{E}_{2}$ (or \emph{destroyed} in the opposite directions of the curve crossings).  We note that the equilibrium $\mathcal{E}_{3}$ does not bifurcate further (i.e., into something else different from  $\mathcal{E}_{1,2}$).

We see the different roles played by the two bifurcations, saddle-node and pitchfork, in creating new solutions and also in transfiguring one into another. In this way, one witnesses the transitions described by the smooth transformations,
\be
\mathcal{E}_3\leftrightsquigarrow\mathcal{E}_1,\quad\textrm{on the}\,\,\tau_+\textrm{-branch},
\ee
and,
\be
\mathcal{E}_3\leftrightsquigarrow\mathcal{E}_2,\quad\textrm{on the}\,\,\tau_-\textrm{-branch},
\ee
(cf. \cite{cot24a}, Sect. 5.4.2,3 for complete proofs of this result).

The $(-)$-case (cf. \cite{cot24a}, Sect. 6) is very different qualitatively than the previous one. At zero parameter the dynamics is characterized by the existence of closed trajectories (only!). There are again three parameter-dependent equilibria which we again denote by $\mathcal{E}_{1,2,3}$. A main difference is that here $\mathcal{E}_{3}$ is a node rather than a saddle of the previous case. There are saddle-node and pitchfork bifurcations of the equilibria $\mathcal{E}_{1,2}$ like before, only that the latter are to a node (not a saddle), which is a sink (resp. source) for $\mu_2>0$ (resp. $<0$).

There is however, one extra possibility for the equilibrium $\mathcal{E}_{3}$, namely, to become a centre acquiring a pair of purely imaginary eigenvalues, and bifurcate in a degenerate Hopf bifurcation thus creating an infinity of closed orbits. It is here that the effects of the cubic term present in the versal unfolding of this case enter and create  a stabilization mechanism that leads to a unique stable linit cycle for the versal dynamics of the FL universes. Upon further parameter variation, one sees that this cycle disappears by becoming invisible, cf. \cite{cot24a} and refs. therein, where a more detailed description of the various bifurcations is given.

\subsubsection{The FL cusp}\label{type-1}
Using similar methods to those in the last two subsections, one is led to the construction of the bifurcation diagrams in this case too and the resulting versal dynamics. The bifurcation diagrams that describe the FL cusp (\ref{va-tot}) correspond to the \emph{quadratic} Bogdanov-Takens versal family,
\begin{equation}\label{va-tot-A}
    \begin{split}
    \dot{x}&=y,\\
    \dot{y}&=\mu_1+\mu_2 y+x^2\pm xy,
    \end{split}
\end{equation}
are shown in Fig. \ref{fig-6}, \ref{fig-7},  and may be found in various sources, for instance, \cite{gh83}, p. 366, \cite{wig}, pp. 444-5, or \cite{ar94}, pp. 24-5, for the $(+)$-case (the  $(-)$-case is easily obtained by symmetry).

\begin{figure}[htp]
\centering
\includegraphics[width=\textwidth]{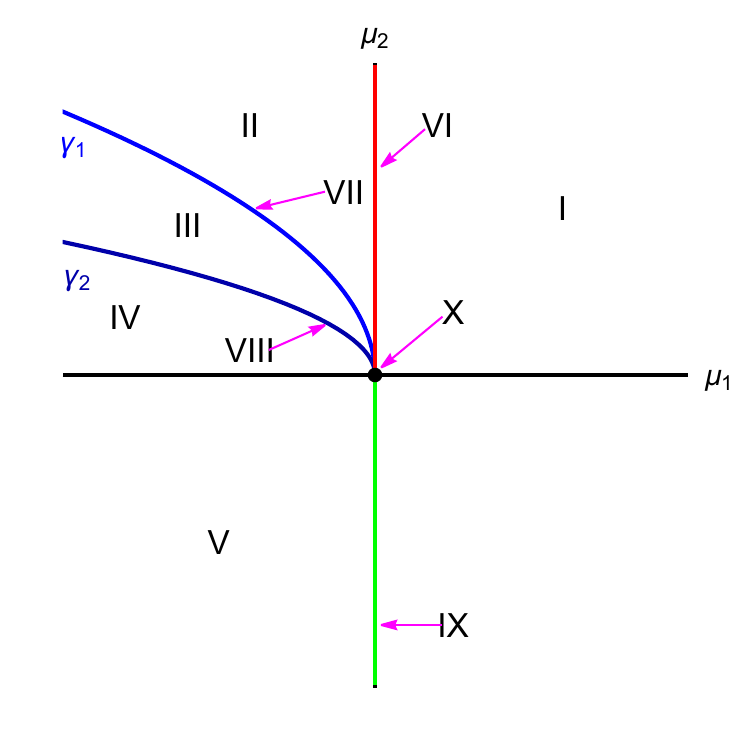}
\caption{The stratified parameter space, part of the bifurcation diagram of FL cusp, positive moduli case. }\label{fig-6}
\end{figure}
\begin{figure}[htp]
\centering

  \includegraphics[width=0.32\textwidth]{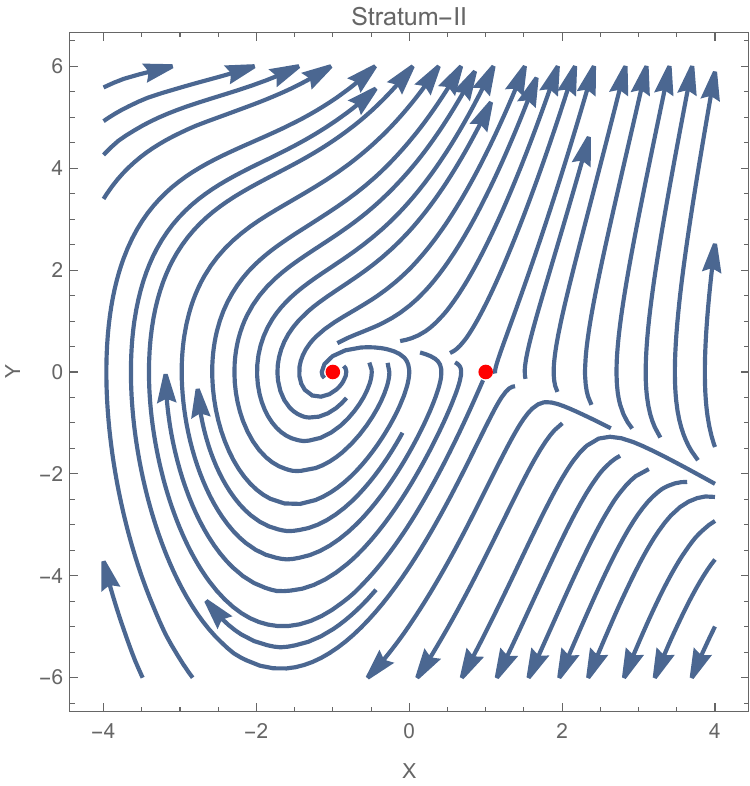}\hfill
  \includegraphics[width=0.32\textwidth]{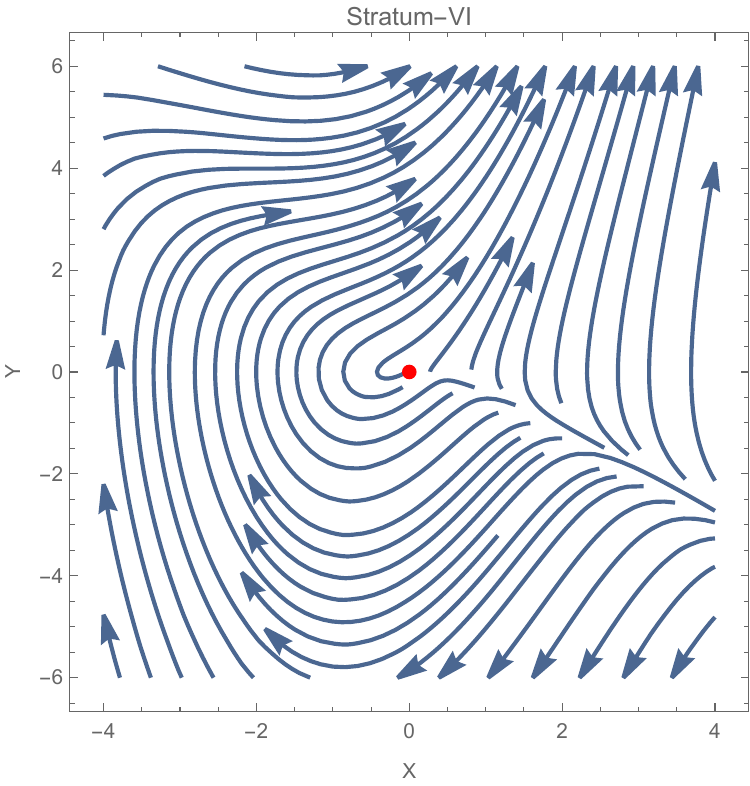}\hfill
  \includegraphics[width=0.32\textwidth]{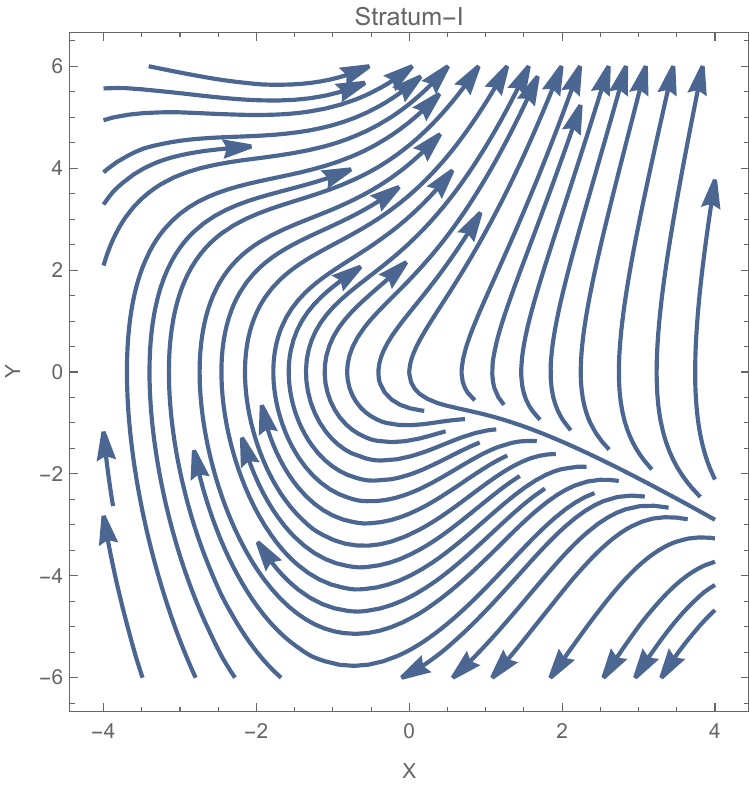}

  \medskip

  \includegraphics[width=0.32\textwidth]{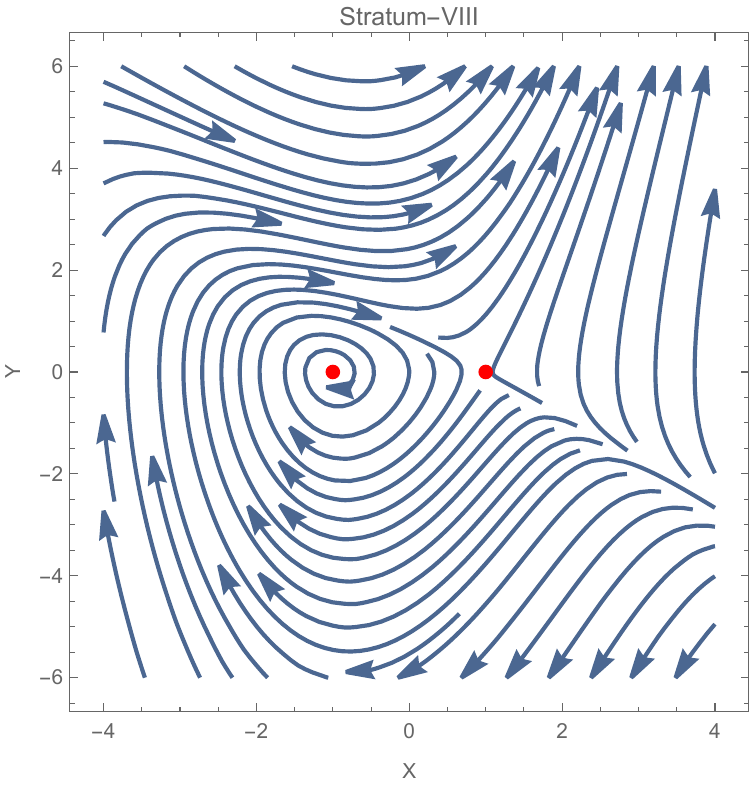}\hfill
  \includegraphics[width=0.32\textwidth]{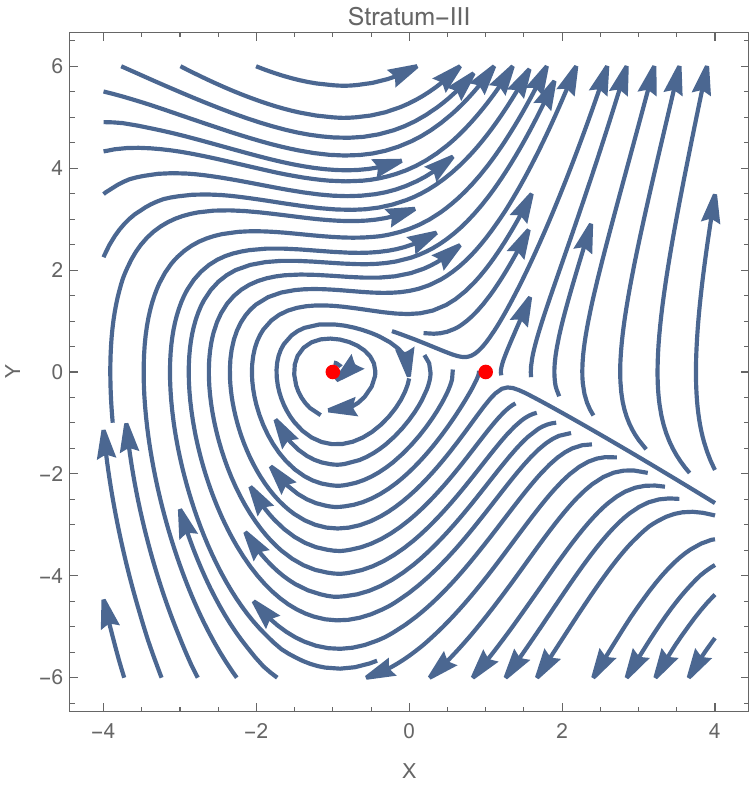}\hfill
  \includegraphics[width=0.32\textwidth]{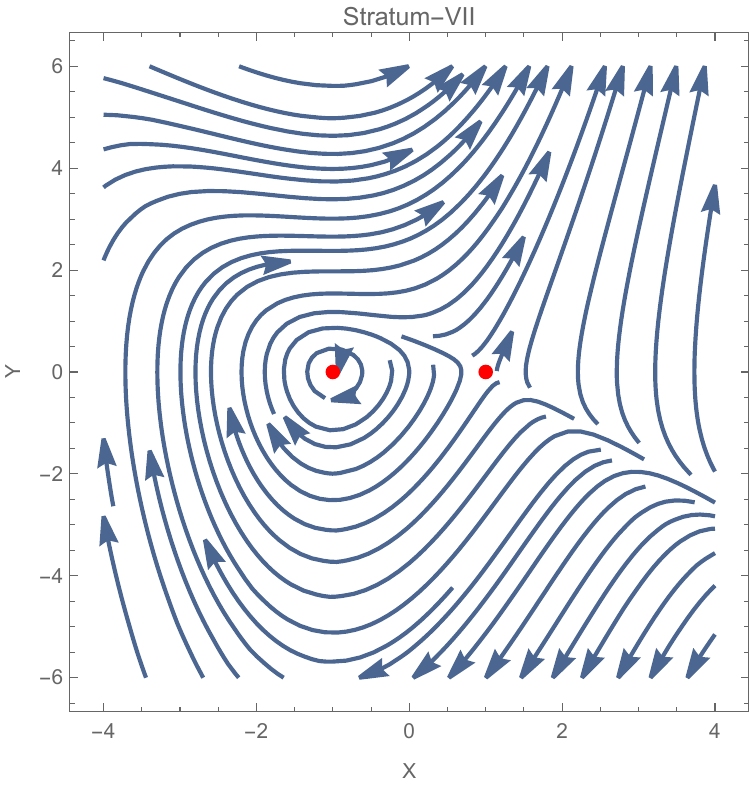}

  \medskip

  \includegraphics[width=0.32\textwidth]{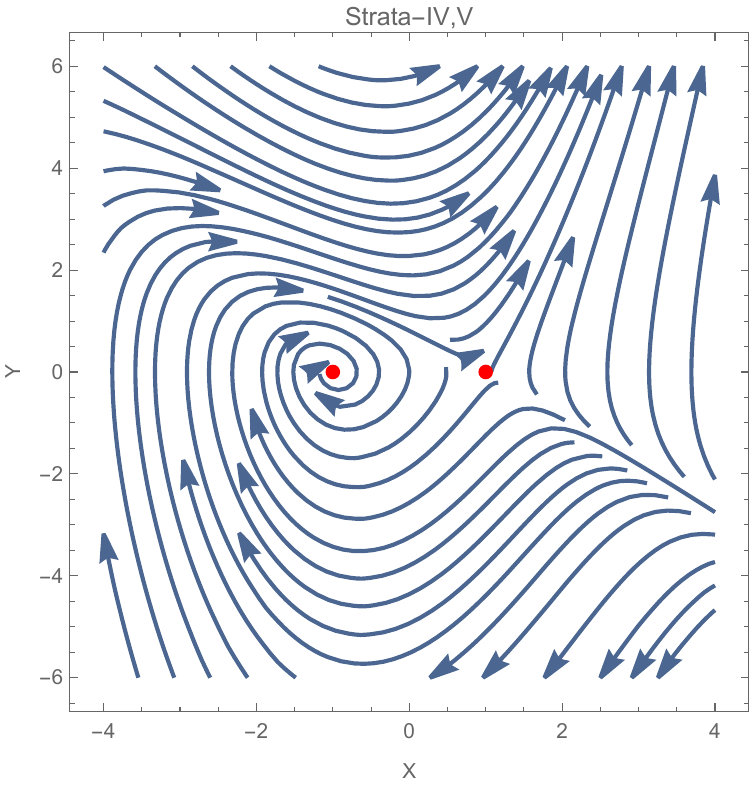}\hfill
  \includegraphics[width=0.32\textwidth]{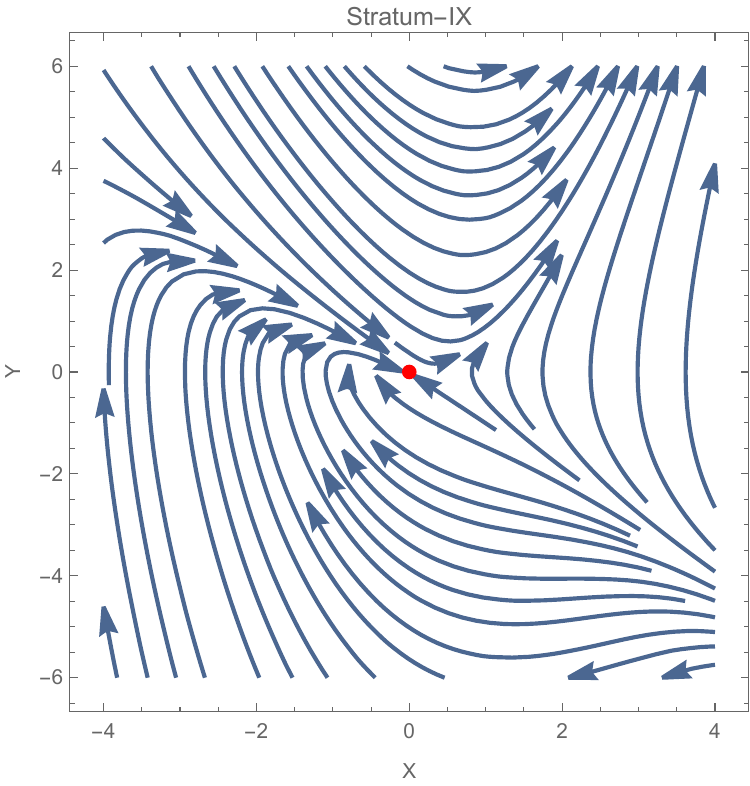}\hfill
  \includegraphics[width=0.32\textwidth]{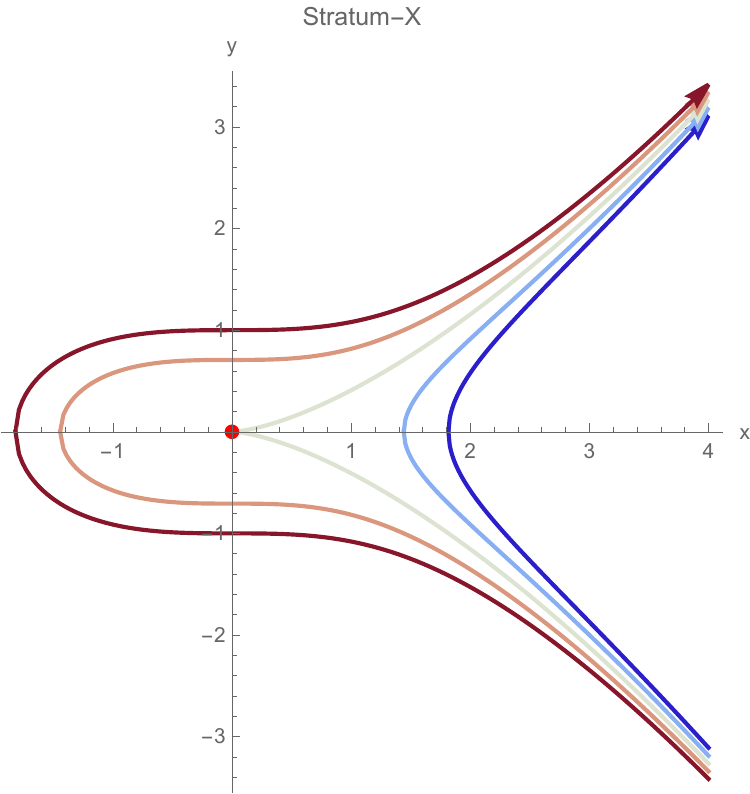}

  \caption{All phase diagrams of the nine strata of Fig. \ref{fig-6}.}
  \label{fig-7}
\end{figure}

There are three important bifurcations for the FL-cusp (i.e., quadratic Bogdanov-Takens) problem (\ref{va-tot-A}) in this case, namely,
\begin{itemize}
  \item A saddle-node bifurcation along the $\mu_2$-axis. This leads to the creation of new parameter-dependent equilibria (when $\mu_1>0$, there are no such equilibria) apart from the origin.
  \item A Hopf bifurcation along the curve $H$ (no relation to the Hubble parameter) given by $\mu_1=-\mu_2^2$. This leads to the appearance of a limit cycle.
  \item A global bifurcation associated to an orbit homoclinic to one of the new equilibria (the saddle) leading to the destruction of the limit cycle.
\end{itemize}

All these features of the bifurcation diagram are insensitive to the inclusion of higher-order terms in the unfolding. We further note that the limit cycle is unstable in the $(+)$-case, stable in the $(-)$-case. The parameter space curve $\gamma_2$, responsible for the creation of the limit cycle in the central diagram of row-2 in Fig. \ref{fig-7}, starts at the origin (i.e., at de Sitter space or the ESU), and is given by $\mu_1=-(49/25)\mu_2^2$. As we trace this curve towards the origin, the homoclinic loop orbit shrinks and disappears, a reason why such a phenomenon in absent in linearized stability studies. Also no  limit cycles exist outside the region III between the two curves $\gamma_1,\gamma_2$. Note also the cusp with `apex' at the fixed point in the stratum X in Fig. \ref{fig-7}, with the two separatrices given by Andronov's theorem 5.3, Case $A_1$(2).

\subsubsection{Codimension-3}\label{type-2}
As it is perhaps expected from Andronov's theorem, Sect. \ref{ant-sect}, there are three topologically inequivalent cases associated with the family (\ref{va-tot4-A}), the \emph{saddle case}  corresponding to the $(+)$-case in Fig. \ref{fig-8} (left), the \emph{elliptic (domain) case} corresponding to the $(-)$-case with $B\geq 2\sqrt{2}$ in Fig. \ref{fig-8} (right), and the \emph{focus case} for the $(-)$-case with $0<B<2\sqrt{2}$ Fig. \ref{fig-9}.
The bifurcation diagrams for the codimension-3 FL cubic, namely (\ref{va-tot5}),
\begin{equation}\label{va-tot4-A}
    \begin{split}
    \dot{x}&=y,\\
    \dot{y}&=\mu_1+\mu_2 x+\mu_3y \pm x^3 +B xy,
    \end{split}
\end{equation}
can be found  in \cite{du1}, figs. 3, 4 for the saddle case, figs. 7, 8 for the focus case, and figs. 13, 14 for the elliptic case (we note that the parameter space in all three cases has a conic structure and the shown diagrams are projections on different faces of a cube centered at the origin of the parameter space).  This is the dynamically richest case among all versally unfolded FL cosmologies, and we conclude this Section by giving a few hints about the dynamics involved in this case, closely following Refs. \cite{du1,du2}.
\begin{figure}[htp]
\centering
  \includegraphics[width=0.5\textwidth]{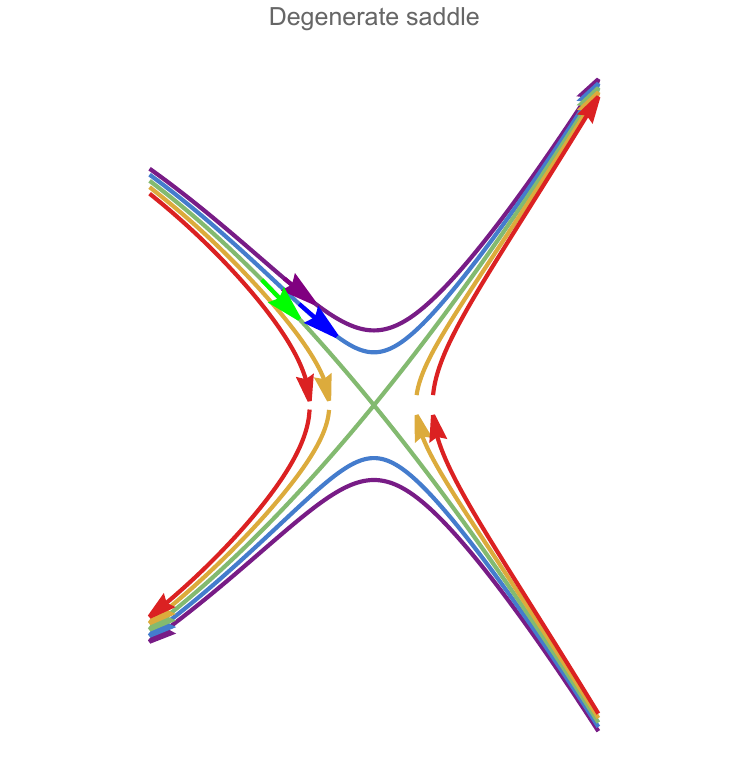}\hfill
  \includegraphics[width=0.5\textwidth]{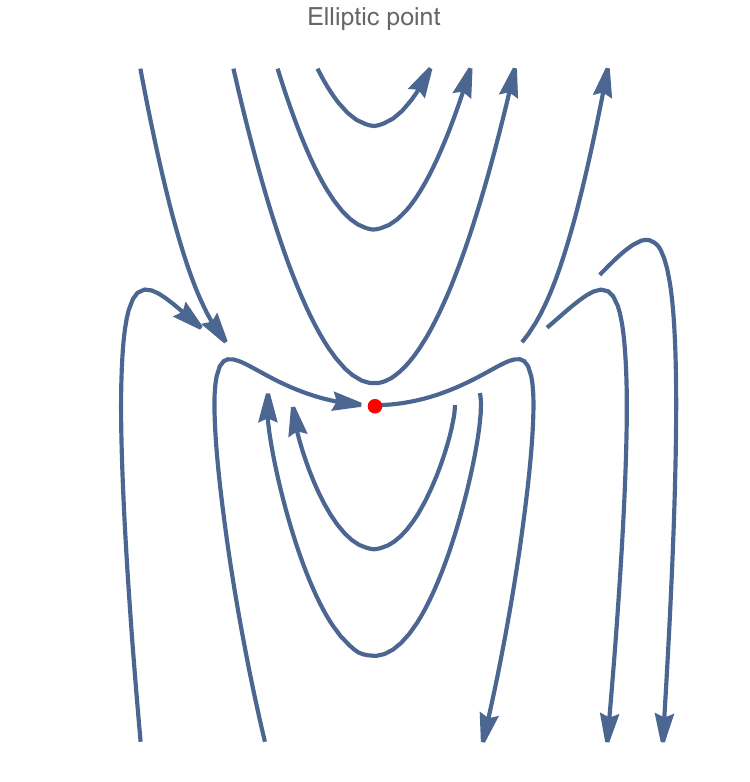}
\caption{The degenerate saddle and elliptic cases}\label{fig-8}
\end{figure}

\begin{figure}
\centering
  \includegraphics[width=0.5\textwidth]{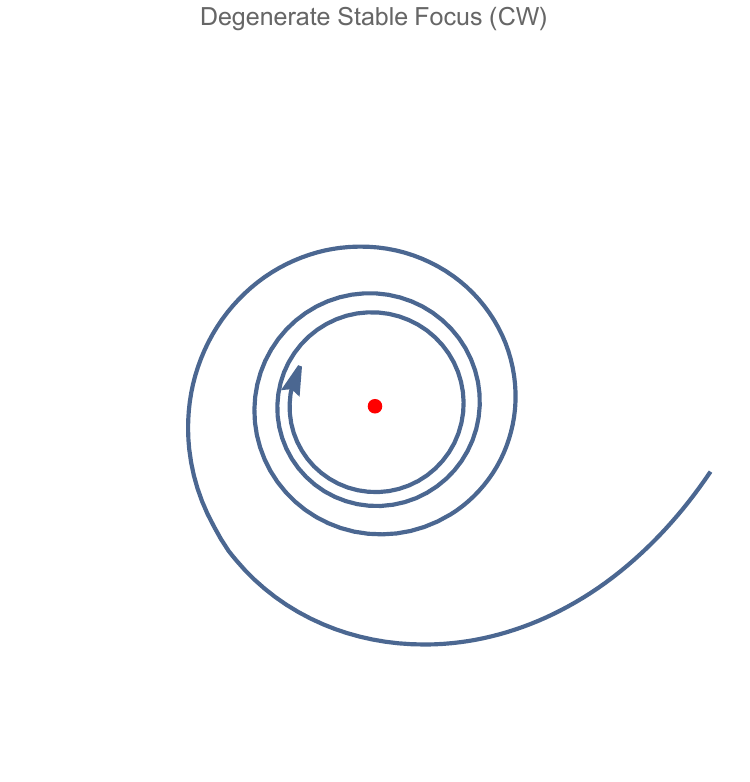}\hfill
  \includegraphics[width=0.5\textwidth]{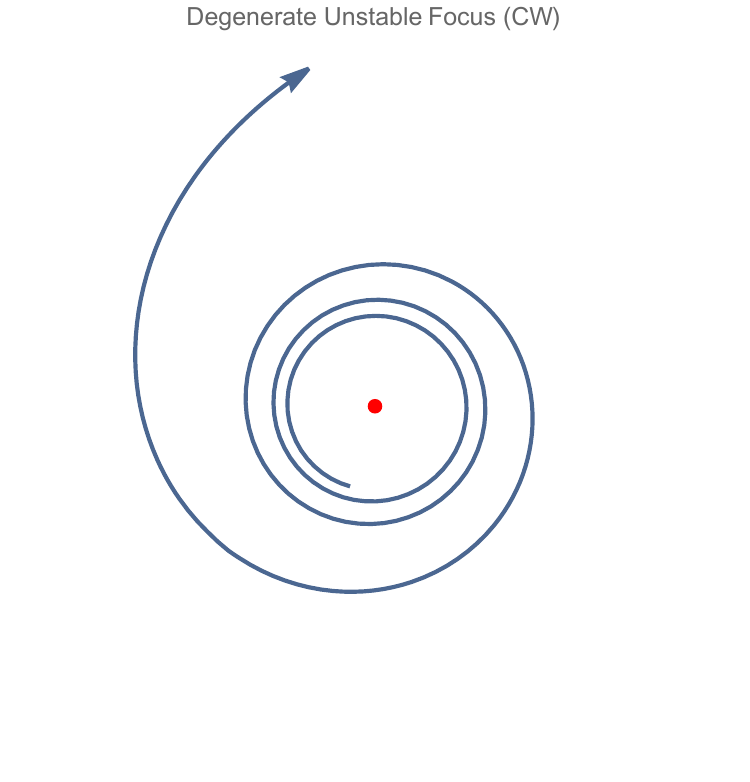}
\caption{The degenerate focus case in the clockwise (CW) direction. }\label{fig-9}
\end{figure}
\begin{figure}
\centering
  \includegraphics[width=0.32\textwidth]{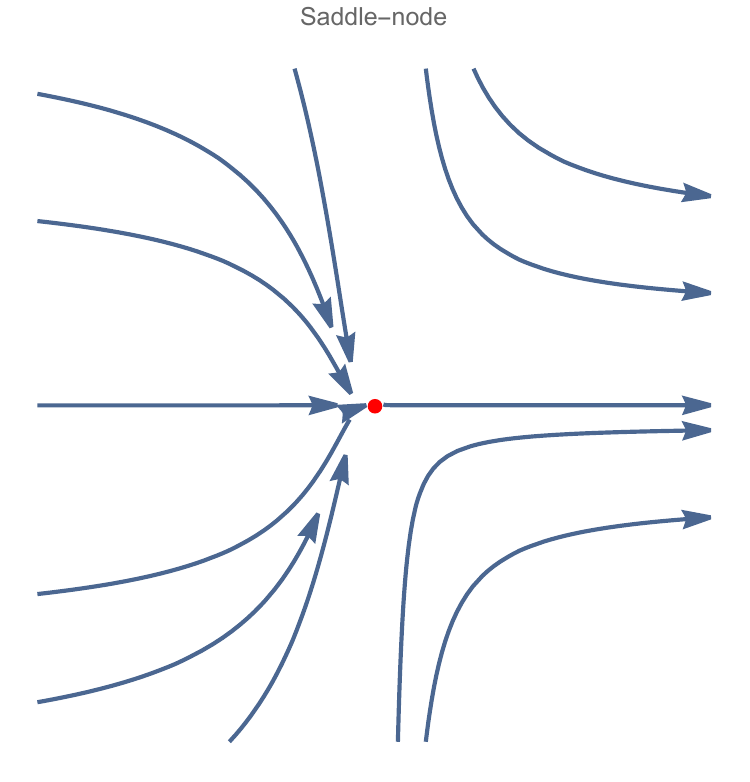}\hfill
  \includegraphics[width=0.32\textwidth]{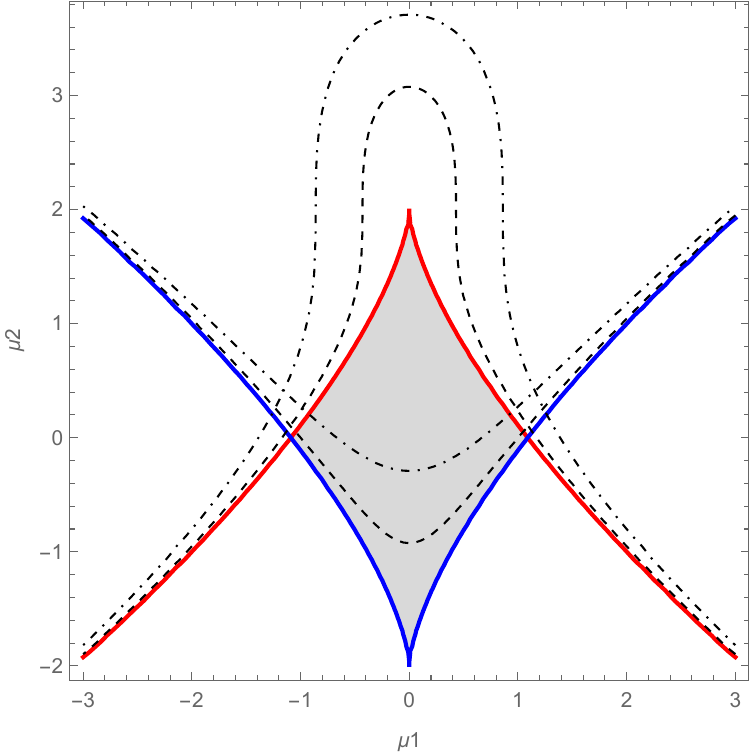}
  \includegraphics[width=0.32\textwidth]{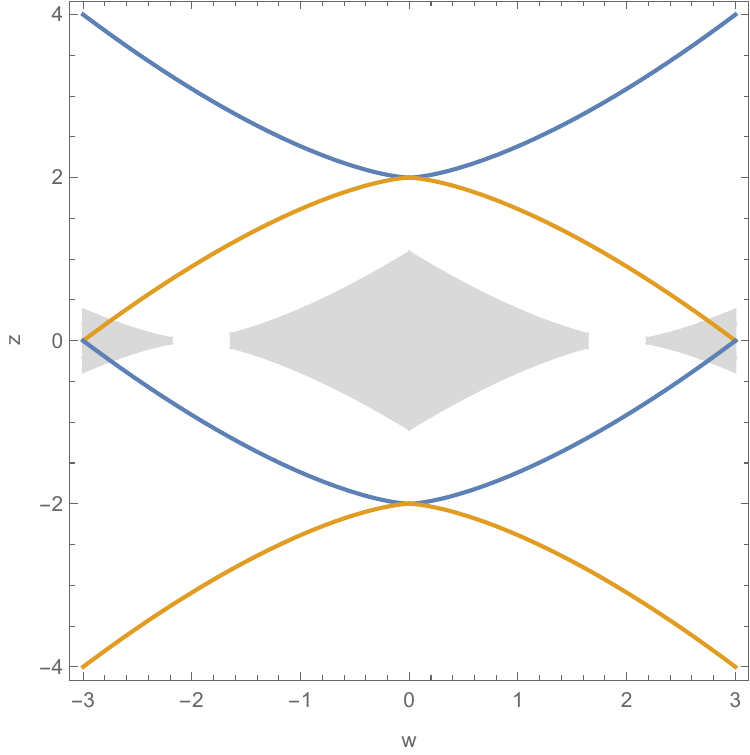}
\caption{A saddle-node-point structure (left), the two cuspidal `vertical lips' enclosing the rhombic region (middle), and the resulting `parallel' saddle-node lines separating the bifurcation diagrams of the FL  cubic into an interior and an exterior region (right).}\label{fig-10}
\end{figure}

The bifurcation diagrams in all cases  are very complex and dependent on specific features that appear in each of the  three cases of the problem. But there is a common theme in all three cases, namely, the existence of the cuspidal `saddle-node' curves shown in \ref{fig-10}, middle and right. The condition  for the equilibria of the system (\ref{va-tot4-A}) is the cubic equation,  $\mu_1+\mu_2 x \pm x^3=0$. For the $(+)$-case, the discriminant is  $\Delta_+=-4\mu_2^3 -27\mu_1^2$, while $\Delta_-=4\mu_2^3 -27\mu_1^2$ for the $(-)$-case in the cubic equation. The two bifurcation sets, $\Delta_\pm=0$, for the two pairs of cusps, stratify the parameter space into an \emph{internal region} where there are three reals roots to \emph{both} cubic  equations (and so in all three cases: saddle, focus, and elliptic),  and an \emph{external region} otherwise.

The two inequalities carve out the  cusp-shaped, `rhombic',  region in parameter space shown in Fig. \ref{fig-10} middle. While the two bifurcation sets given by the cusps meet at the origin of the parameter space $(\mu_1,\mu_2)$, the thought strikes one that they can be arranged so that to show more clearly their common region with the three real roots. This can be done in the following way, cf. Fig. \ref{fig-10} middle: We `lift' the two cusps apart by a vertical shift $\Delta$, so that $\pm x^3$-cubic solves the equation, $\Delta_\pm(\mu_1,\mu_2\mp\Delta)=0$,  that is we define  so that the upper red cusp (from the $(+)$-case) sits at $\mu_2=+\Delta$, the lower blue cusp (from the $(-)$-case) sits at $\mu_2=-\Delta$. The gray `vertical-lips'\footnote{This terminology comes from \cite{du2}.} region in between is exactly the set of $(\mu_1,\mu_2)$ where both cubics have three real roots (i.e. $\Delta_+\leq 0$ and $\Delta_-\geq 0$). In this way,  the two $\Delta=0$-cusps no longer meet at $\mu_2=0$ but at $\mu_2=+\Delta$ (red) and $\mu_2=-\Delta$ (blue).  Also shown are  some parallel iso-discriminant (level-)curves $\{(\mu_1,\mu_2)|\Delta(\mu_1,\mu_2)=C\}$, i.e., contours of constant discriminant of the  $\pm$-cubics.

Our construction of the `lens' (or `vertical lips') region of Fig. \ref{fig-10} middle is completely equivalent to that followed in Refs. \cite{du1,du2}. In these works, the authors choose the parabolic  parametrization, drawing instead the `semi-cubic parabolas' $\displaystyle \mu_1 = \pm \sqrt{\frac{4}{27}}\;\bigl|\mu_2 \mp \Delta\bigr|^{3/2}$ directly, whereas in our `rhombic approach' we shade the implicit discriminant inequalities $\displaystyle
\Delta_+(\mu_1,\mu_2 - \Delta) = 4(\mu_2 - \Delta)^3 + 27\,\mu_1^2 \;\le\; 0,\,\,
\Delta_-(\mu_1,\mu_2 + \Delta) = 4(\mu_2 + \Delta)^3 - 27\,\mu_1^2 \;\ge\; 0
$. The comparison of the two approaches, and their complete equivalence, is shown in Fig. \ref{fig-10} right.

The following further aspects of the bifurcational dynamics of (\ref{va-tot4-A}) proved in Refs. \cite{du1,du2} are very important:
\begin{itemize}
  \item The versal unfolding (\ref{va-tot4-A}) is shown to be topologically equivalent to the general topological normal forms considered in Ref. \cite{du2} by the method of \emph{principal rescaling}, cf. \cite{du1,du2}.
  \item There are further bifurcation curves in the internal and external regions apart from the saddle-node curves discussed above. These include Hopf bifurcation curves, Bogdanov-Takens points,  saddle connections (homoclinic or heteroclinic), and various loops and tangencies of more complicated types.  There are codim-1 strata (surfaces) and codim-2 strata (curves) transversal to the spheres $S:\mu_1^2+\mu_2^2+\mu_3^2=\varepsilon^2$, for $\varepsilon$ small. The bifurcation sets intersect such spheres $S$.
  \item The system has \emph{at most one limit cycle} in the elliptic case, and in suitable parameter regions in the other two cases.

  \item There are further structures \emph{at infinity}, for example, limit cycles surrounding the three singular points in the focus case, and a separatrix from the saddle point to infinity in the elliptic case.
\end{itemize}
The reduction for the FL equations effected above leading to  the versal unfolding (\ref{va-tot4-A}) together with the results just mentioned lead to a novel cosmological dynamics having little in common not only with the known cosmological solutions  obtained by an analysis of the original FL equations, but also with all the previous three bifurcation cases considered  in this paper. The main novelties here  are related to  new equilibria and cycles, their different  creation and annihilation properties, as well as the several possible distinct ways for their transfigurations.

\subsection{Rough solutions}\label{rough}
Let us finally discuss the meaning of the Friedmann equation (\ref{ds3}) in connection with the  versal  unfoldings (\ref{va-tot}), (\ref{va-tot4}), (\ref{va-3rd}), and (\ref{va-2/3}). As we showed in previous Sections, in any of the versal unfoldings, all relations between the unfolding parameters and the phase portrait variables $x,y$ are always smooth, even analytic. That is the versal families are smooth functions of their arguments. This is evident from the constructions and final forms of all versal families we obtained.

The new coordinates $x,y$, however, are,  as \emph{solutions} of the versal equations,  \emph{non-smooth} functions of  the distinguished parameters $\gamma, \Lambda$, and the  auxiliary parameter $\mu$ being a 2-, or 3-dimensional vector with coordinates $\mu_1, \mu_2$, and $\mu_3$\footnote{The reader should not confuse this with the smoothness of each phase portrait or of their metamorphoses.}. Therefore if we write the right-hand-sides of the versal unfolding equations (\ref{va-tot}), (\ref{va-tot4}), (\ref{va-3rd}), and (\ref{va-2/3}) in the symbolic form $(f,g)$, where,
\be\label{hr}
f,\,\,g=f(H,\rho,\gamma,\Lambda, \mu_1, \mu_2, \mu_3),\,\,g(H,\rho,\gamma,\Lambda, \mu_1, \mu_2, \mu_3),
\ee
to signify their dependence on the variables $H,\rho$ and the various parameters, we can comment on the smoothness properties of the final solutions in terms of the original $H,\rho$ variables. With regard to the equations,
\be \label{rou1}
f=g=0,
\ee
in the variables (\ref{hr}), it is similarly rarely possible to \emph{further solve} for $H$ or $\rho$ - let alone the scale factor $a$ - to express them in \emph{a smooth or even a continuous manner} in terms of the parameters $\gamma,\Lambda, \mu_1, \mu_2, \mu_3$. This is so because of the powers present there, some kind of non-smooth, or discontinuous  behaviour will always be present in the functions $H=H(\gamma,\Lambda, \mu_1, \mu_2, \mu_3)$, and $\rho=\rho(\gamma,\Lambda, \mu_1, \mu_2, \mu_3)$, or in their derivatives\footnote{We note here the subtle fact that all \emph{standard} solutions of the FL equations are \emph{smooth as functions of the} $\gamma,\Lambda$, as it clearly follows from the Friedmann equation (\ref{ds3}) and by inspecting  the  solutions (\ref{ds-eq}), (\ref{esu-eq}), the $\gamma,\Lambda$ being constants.}.

Hence, we arrive at the important conclusion that typically our solutions obtained from the versal unfoldings (\ref{va-tot}), (\ref{va-tot4}), (\ref{va-3rd}), and (\ref{va-2/3}) will be non-smooth. In addition, this will also lead to behaviours which are typically inhomogeneous  in the functions $H, \rho$ and the scale factor $a$ because of the following effect. In general, derivatives of these functions with respect to the unfolding parameters are of a `non-local' character in terms of the phase points because they relate phase portraits at two different $\mu$-values, which are typically topologically inequivalent due to the bifurcations. For instance, this will be the case in parameter space when passing from a stratum corresponding  to  one equilibrium to another corresponding to two equilibria, for example in a saddle-node bifurcation.

It is typical in any bifurcation to expect a jump, a discontinuity, or in general non-smooth behaviour due to the change in the number of solutions upon parameter change. This effect therefore leads to expect a `rough' character for all our solutions, in the sense that although we work with smooth, even analytic,  relationships between the variables as it is clearly evident in the constructions of all versal unfoldings, one eventually discovers a lot of non-smoothness and the resulting inhomogeneous behaviour at the end parameter-dependent solutions.

\section{Discussion and outlook}\label{dis}

In this paper, we have taken into account two somewhat subtler aspects of the mathematical structure of the FL equations and based on these we began to deploy  some of their consequences for the nature and the evolution of the unfolded solutions. The first aspect of the structure of the FL equations we exploited here is the fact that $\gamma$ and  $\Lambda$ are two distinguished parameters of the problem which play totally different roles in the evolution of the universe. Secondly,  the FL equations are a degenerate system whose linear part contains two zero eigenvalues (and is therefore either the nilpotent or else the zero matrix). We have shown that the nonlinear terms in the FL equations determine precisely four main degenerate cases to which the original dynamics precisely  reduces. A common characteristic of the unfolded evolution in all four cases is the presence of a number  of novel equilibria, totally distinct from the standard ones (i.e., the Einstein static universe and de Sitter space). These equilibria are parameter-dependent and totally absent  when the parameters are zero.

The presence of symmetry in the FL equations, in particular the $\mathbb{Z}_2$ symmetry $x\to\epsilon x,y\to\delta y, \epsilon,\delta=\pm 1$, and the time symmetry $t\to-t$ also play  important roles in this respect. In fact, the whole dynamics is initially determined by the combined effects of $\gamma$ and the presence or absence of the time symmetry, and the question arises as to how to study these systems without spoiling the degeneracies completely. If the FL equations were an isolated `hyperbolic' system then such degenerate, nongeneric  cases could have been removed by a small perturbation moving the system  into the set of generic cases.

However, in the case of the FL equations a small perturbation of the parametric family (the latter represented say by a curve parameterized by the $\Lambda$) intersects the surface containing the $\Lambda =0$ system transversally (i.e., at nonzero angle) at a nearby point of the degenerate surface. So although every member of the family can be made generic by a small perturbation, it is not possible to achieve this for the whole family at once, as a nongeneric case is avoided at some parameter value only to recur at another.

Therefore in this paper we have striven to calculate the corresponding  codimensions of the degenerate FL systems. We found them equal to 2 or 3 depending on the type of degeneracy, and we have also proceeded to find their bifurcations (both steps are necessary, cf. \cite{ar83}, Section 29 for additional motivation on this process). This in turn leads to a very rich set of bifurcation diagrams for the versal unfoldings of the FL equations. These unfoldings describe and exhaust all possible perturbations of the original FL equations and their solutions sets - the bifurcation diagrams - describe the possible topologically inequivalent kinds of evolution completely.  The resulting solutions are typically non-smooth and this implies  global inhomogeneities present in the character and the  physics of the versal unfoldings.

We have shown the presence of bifurcations that lead to the creation or annihilation of new equilibria other than the de Sitter space or the ESU cases. There are other bifurcations which transfigure solutions to others in a very precise manner. There are also isolated closed orbits (i.e., limit cycles) in the phase spaces associated with these new structures implied by the FL equations. The cycles bifurcate differently than  equilibria depending on codimension. There are also other kinds of connecting orbits between different equilibria and cycles, for example, saddle connections and saddle-node loops which are phenomena of codimension 2, etc. We shall return to this problem in the future with a more detailed description of some of these effects.

Another problem which may lead to interesting results is the extension of the present approach to other contexts (cf. Ref. \cite{cot25a}, for a more thorough discussion of this issue and gravitational bifurcation methods in general). We shall restrict presently to only a few comments about one of the possible ways to study this problem, namely in the context of modified gravity. Suppose one is interested in developing versally unfolded FL cosmology in a modified gravity context, say for some $f(R)$ theory or in a Brans-Dicke context with a fluid with parameter $\gamma$ and a cosmological constant $\Lambda$. We assume that one is able to find a smooth series of changes of the variables in the original equations in that context recasting them near some equilibrium solution to the form of a nonlinear system  for some new variables $X=(x,y)$ of the form,
\be \label{deg-sys-mod}
\dot{X}
=
\left(
  \begin{array}{cc}
    0 & 1\\
    0 & 0 \\
  \end{array}
\right)
X
+
\cdots ,
\ee
where the dots mean \emph{any} nonlinear terms in $X$ (obtained say from an expansion of the vector field $f(x,y)$ defining the problem around the origin). In this case, identical conclusions to those found here in Sections 3-7 will apply in the modified gravity context, and the whole problem there will be precisely reduced to the one discussed here. However, the present results are expected to be generically violated in such contexts, either in the Jordan or Einstein frame formulations of a theory in the context we are now discussing. The main issue is related to the conformal equivalence of these theories to Einstein gravity \emph{with a self-interacting scalar field}. The present of such a scalar field  is expected to disrupt the present conclusions, because the dimensionality of the system  will be increased to become at least three and then one expects the influence of new interactions between the various modes of bifurcation to become important. This has not been done in a gravity context even for the simplest genre of self-interaction. While one may develop gravitational bifurcations independently of frame, another issue arises as to the behaviour of bifurcations under a conformal transformation taking those of one conformal frame to those of another. Of course, one does not  expect to find  a uniform kind of bifurcational behaviour along these lines  \emph{across} the space of modified theory lagrangians. We believe that these remarks will be useful in a wider context of investigations for gravitational bifurcation theory than in the problem developed presently.

\addcontentsline{toc}{section}{Acknowledgments}
\section*{Acknowledgments}
The author is grateful to two anonymous referees for their constructive comments which led to an improvement of this work. This research  was funded by RUDN University,  scientific project number FSSF-2023-0003.

\end{document}